\documentclass[11pt,letter]{article}
\usepackage{amsmath}
\usepackage{xspace}
\usepackage{todonotes}
\usepackage{amsthm}

\usepackage[utf8]{inputenc}
\usepackage{indentfirst,verbatim} 
\usepackage{listings} 
\usepackage[normalem]{ulem}
\usepackage{soul}
\usepackage{amsfonts}
\usepackage{amsthm}
\usepackage{amssymb}
\usepackage{enumitem}
\usepackage{hyperref}

\usepackage{amsmath,amsfonts,amsthm}
\usepackage{mathbbol}
\usepackage{amsthm}
\usepackage{graphicx,color}
\usepackage{boxedminipage}
\usepackage[ruled,boxed,linesnumbered]{algorithm2e}
\usepackage[noend]{algpseudocode}
\usepackage{framed}
\usepackage{thmtools}
\usepackage{thm-restate}
\usepackage{xspace}
\usepackage{todonotes}
\usepackage{pgfplots}
\usetikzlibrary{calc}
\usetikzlibrary{shapes}
\usetikzlibrary{arrows.meta}
\usepackage{colortbl}
\usepackage{tcolorbox} 

\usepackage{xifthen}
\usepackage{tabularx}
\usepackage{capt-of}
\usepackage{thmtools}
\usepackage{subcaption}
\usepackage{booktabs}
\usepackage{calc}
\usepackage{cleveref}

\usepackage[numbers]{natbib}

\usepackage[margin=1in]{geometry} 

\theoremstyle{plain}

\newtheorem{observation}{Observation} 
\newtheorem{corollary}{Corollary} 
\newtheorem{lemma}{Lemma} 
\newtheorem{claim}{Claim} 
\newtheorem{proposition}{Proposition}

\theoremstyle{definition}

\newcommand{\cqed}{\renewcommand{\qedsymbol}{$\lrcorner$}\qed}

\newcommand{\pname}{\textsc}
\newcommand{\polyn}{n^{\Oh(1)}}
\newenvironment{claimproof}{\medskip\noindent \emph{Proof of Claim~\theclaim.}  }{\hfill\cqed\medskip}

\newlength{\RoundedBoxWidth}
\newsavebox{\GrayRoundedBox}
\newenvironment{GrayBox}[1]%
{\setlength{\RoundedBoxWidth}{.93\textwidth}
	\def\boxheading{#1}
	\begin{lrbox}{\GrayRoundedBox}
		\begin{minipage}{\RoundedBoxWidth}}%
		{   \end{minipage}
	\end{lrbox}
	\begin{center}
		\begin{tikzpicture}%
			\node(Text)[draw=black!20,fill=white,rounded corners,%
			inner sep=2ex,text width=\RoundedBoxWidth]%
			{\usebox{\GrayRoundedBox}};
			\coordinate(x) at (current bounding box.north west);
			\node [draw=white,rectangle,inner sep=3pt,anchor=north west,fill=white] 
			at ($(x)+(6pt,.75em)$) {\boxheading};
		\end{tikzpicture}
\end{center}}     

\newenvironment{defproblemx}[2][]{\noindent\ignorespaces%
	\FrameSep=6pt%
	\parindent=0pt%
	\vspace*{-1.5em}
	\ifthenelse{\isempty{#1}}{%
		\begin{GrayBox}{#2}%
		}{%
			\begin{GrayBox}{#2 parameterized by~{#1}}%
			}
			\begin{tabular*}{\textwidth}{@{\hspace{.1em}} >{\itshape} p{1.8cm} p{0.8\textwidth} @{}}%
			}{
			\end{tabular*}%
		\end{GrayBox}%
		\ignorespacesafterend
	}  
	
	\newcommand{\defparproblema}[4]{
		\begin{defproblemx}[#3]{#1}
			Input:  & #2 \\
			Task: & #4
		\end{defproblemx}
	}%

	\newcommand{\Oh}{\mathcal{O}}

	\newcommand{\probECC}{\pname{Edge Clique Cover}\xspace}
	\newcommand{\probECP}{\pname{Edge Clique Partition}\xspace}
	\newcommand{\probECCshort}{\pname{ECC}\xspace}
	\newcommand{\probECPshort}{\pname{ECP}\xspace}
	\newcommand{\probAECC}{\pname{Annotated Edge Clique Cover}\xspace}

	\newcommand{\probECPIS}{\pname{Edge Clique Partition Above Independent Set}\xspace}
	\newcommand{\probECCIS}{\pname{Edge Clique Cover Above Independent Set}\xspace}
	\newcommand{\probECPISshort}{\pname{\probECPshort/$\alpha$}\xspace}
	\newcommand{\probECCISshort}{\pname{\probECCshort/$\alpha$}\xspace}
	
	\newcommand{\probAECCshort}{\pname{Annotated ECC}\xspace}	
	\newcommand{\probVCC}{\pname{Vertex Clique Cover}\xspace}

	\newcommand{\bfx}{\mathbf{x}}
	\newcommand{\bfy}{\mathbf{y}}

	\DeclareMathOperator{\operatorClassP}{{\sf P}}
	\newcommand{\classP}{\ensuremath{\operatorClassP}}
	\DeclareMathOperator{\operatorClassNP}{{\sf NP}}
	\newcommand{\classNP}{\ensuremath{\operatorClassNP}}

	\DeclareMathOperator{\operatorClassFPT}{{\sf FPT}\xspace}
	\newcommand{\classFPT}{\ensuremath{\operatorClassFPT}\xspace}
	\DeclareMathOperator{\operatorClassW}{{\sf W}}
	\newcommand{\classW}[1]{\ensuremath{\operatorClassW[#1]}}
	\DeclareMathOperator{\operatorClassParaNP}{{\sf Para-NP}\xspace}
	\newcommand{\classParaNP}{\ensuremath{\operatorClassParaNP}\xspace}

	\DeclareMathOperator{\ecc}{ecc}	
	\DeclareMathOperator{\ecp}{ecp}	
	\DeclareMathOperator{\aecc}{aecc}	
	\lstnewenvironment{code}{\lstset{language=Haskell,basicstyle=\small}}{}

\title{Edge Clique Partition and Cover Beyond Independence\thanks{The research leading to these results has been
supported by the Research Council of Norway via the project BWCA (grant no. 314528).}}

	\author{
		Fedor V. Fomin\thanks{
			Department of Informatics, University of Bergen, Norway.}\\fedor.fomin@uib.no
		\and
		Petr A. Golovach\addtocounter{footnote}{-1}\footnotemark{}\\petr.golovach@uib.no
		\and
		Danil Sagunov\thanks{
			Saint Petersburg State University, St.\ Petersburg, Russia.
		}\\danilka.pro@gmail.com
		\and 
		Kirill Simonov\addtocounter{footnote}{-2}\footnotemark{}\\kirill.simonov@uib.no
	}
	\date{}

\begin{document}
		
\maketitle	

\begin{abstract}
Covering and partitioning the edges of a graph into cliques are classical problems at the intersection of combinatorial optimization and graph theory, having been studied through a range of algorithmic and complexity-theoretic lenses. Despite the well-known fixed-parameter tractability of these problems when parameterized by the total number of cliques, such a parameterization often fails to be meaningful for sparse graphs. In many real-world instances, on the other hand, the minimum number of cliques in an edge cover or partition can be very close to the size of a maximum independent set~\(\alpha(G)\). 

Motivated by this observation, we investigate \emph{above}-\(\alpha\) parameterizations of the edge clique cover and partition problems. Concretely, we introduce and study   \probECCIS  (\probECCISshort) and 
\probECPIS (\probECPISshort), where the goal is to cover or partition all edges of a graph using at most \(\alpha(G) + k\) cliques, and \(k\) is the parameter. Our main results reveal a distinct complexity landscape for the two variants. We show that \probECPISshort is fixed-parameter tractable, whereas \probECCISshort is NP-complete for all \(k \ge 2\), yet can be solved in polynomial time for \(k \in \{0,1\}\). These findings highlight intriguing differences between the two problems when viewed through the lens of parameterization above a natural lower bound.

Finally, we demonstrate that \(\probECCISshort\) becomes fixed-parameter tractable when parameterized by \(k + \omega(G)\), where \(\omega(G)\) is the size of a maximum clique of the graph $G$. This result is particularly relevant for sparse graphs, in which \(\omega\) is typically small. For $H$-minor free graphs,  we design a subexponential  algorithm of running time $f(H)^{\sqrt{k}}\cdot\polyn$. 
\end{abstract}

 \section{Introduction}
 Covering and partitioning the edges of a graph into cliques are fundamental combinatorial problems that lie at the intersection of combinatorial optimization, graph theory, and complexity theory. In   the  \probECC  (\probECCshort) problem one seeks to cover all edges of a graph with as few cliques as possible, while in 
  \probECP (\probECPshort)  the same objective is pursued under the additional constraint that each edge must lie in exactly one of the chosen cliques. Both problems were extensively studied from various algorithmic perspectives (See \Cref{sec:related} for an overview), in particular from parameterized algorithms.

\probECCshort\ and \probECPshort\ are known to be fixed-parameter tractable (FPT) when parameterized by the number of cliques in the cover or partition~\cite{GrammGHN08,mujuni2008parameterized}. While this parameterization is meaningful for dense graphs, it becomes less relevant for sparse graphs (e.g., graphs with bounded maximum degree or planar graphs), where the size of a clique cover or partition usually scales proportionally with the number of vertices. In such cases, a brute-force algorithm that guesses all possible partitions already runs in FPT time under this parameter choice.

These observations naturally lead us to explore a different parameterization---one that captures a lower bound on the size of the cover or partition. In particular, for any graph \(G\) without isolated vertices, an edge clique cover or partition must contain at least \(\alpha(G)\) cliques, where \(\alpha(G)\) denotes the size of a maximum independent set of \(G\). Empirical studies of \probECCshort\ (see, e.g., \cite[Table~3]{HeviaECCexperiments}) indicate that in many real-world instances, after suitable preprocessing, the number of cliques in the cover is \emph{very} close to \(\alpha(G)\). This observation motivates us to investigate \probECCshort\ and \probECPshort\ in the setting of parameterization \emph{above}~\(\alpha\). Concretely, we initiate the study of the following parameterized problems.

 \defparproblema{\probECPIS (\probECPISshort)}{A graph $G$ without isolated vertices, an integer $k\ge0$.}{$k$}{Find an edge clique partition of $G$ with at most $\alpha(G)+k$ cliques.}

 \defparproblema{\probECCIS (\probECCISshort)}{A graph $G$ without isolated vertices, an integer $k\ge0$.}{$k$}{Find an edge clique cover of $G$ with at most $\alpha(G)+k$ cliques.}

While \probECCshort\ and \probECPshort\ are both FPT when parameterized by the number of cliques, their above-guarantee variants—\probECCISshort\ and \probECPISshort—exhibit notably different and intriguing behaviors.

\subsection{Our Results}
 Our first result establishes parameterized complexity of \probECPIS  (\probECPISshort). It is worth to note that the running time of our algorithm matches the best known running time of the algorithm for \probECPshort~\cite{FeldmanECP20}. Moreover, as we show, any improvement in the running time for   \probECPshort
 would improve the running time in our theorem too. 
   
\begin{restatable}{theorem}{ECPISFPT}
\label{thm:ecpis-fpt}
\probECPISshort can be solved in $2^{\Oh(k^{3/2}\log k)}\cdot n^{\Oh(1)}$ time.
\end{restatable} 
 
We underline that the input of \probECPISshort (or \probECCISshort) does not contain the value of $\alpha(G)$. Interestingly, while computing the maximum independent set is well-known to be intractable,  we show in~\Cref{lem:isecp} that, given an instance  $(G,k)$ of \probECPISshort, in $2^{\Oh(k^{3/2}\log k)}\cdot n^{\Oh(1)}$ time we can either compute $\alpha(G)$ or conclude that $(G,k)$ is a no-instance. This algorithm is used as a subroutine in~\Cref{thm:ecpis-fpt} and, in fact, our algorithm for \probECPISshort either outputs an edge clique cover of size at most $\alpha(G)+k$ together with $\alpha(G)$ or correctly reports a no-instance.

 While both \probECPshort\ and \probECCshort\ are \classFPT\ when parameterized by the number of cliques in the solution, their above-guarantee variants exhibit drastically different behaviors. The following theorem shows that for every fixed integer $k \ge 2$, deciding whether a graph can be covered by at most $\alpha(G) + k$ cliques is \classNP-complete even on prefect graphs for which the independence number can be computed in polynomial time~\cite{GLS1988}. Consequently, this places \probECCISshort in the class of \classParaNP-complete problems.

 \begin{restatable}{theorem}{ECCISNPhard}
 \label{thm:hard}
	For every $k\ge 2,$ \probECCISshort  is \classNP-complete. 
	Furthermore,  the hardness holds even on perfect graphs.
\end{restatable} 

The condition in \Cref{thm:hard} concerning $k$ is tight: for $k < 2$, the problem is solvable in polynomial time. Taken together, \Cref{thm:hard,cor:poly-zero-one} establish a dichotomy result on the complexity of \probECCISshort for all values of $k$.

 \begin{restatable}{theorem}{ECCISPoly}
 \label{cor:poly-zero-one}
	\probECCISshort   
	admits polynomial-time algorithms for $k\in \{0,1\}$.
\end{restatable} 
 
 Despite the intractability of \probECCISshort\ on general graphs, for certain classes of sparse graphs  it is possible to design \classFPT\ algorithms. We summarize these findings in the following theorem.
 
 \begin{restatable}{theorem}{ECCISsparse}
 	\label{cor:sparse}
 	\probECCISshort admits parameterized algorithms with the following running times:
 	\begin{itemize}
 		\item $4^{\binom{\omega}{2}\cdot k}\cdot \polyn$ on graphs with clique number $\omega$, 
 		\item $2.081^{(d-1)k}\cdot\polyn$ on graphs of degeneracy $d\ge 3$,
 		\item $1.619^k\cdot \polyn$ on $2$-degenerate graphs,
 		\item $f(H)^{\sqrt{k}}\cdot\polyn$ on $H$-minor-free graphs.
 	\end{itemize}
 	Neither of $\omega, d, H$ should be given to the corresponding algorithm explicitly.
 \end{restatable}

Complementing this result, we observe that, under the Exponential Time Hypothesis, the dependency on $k$ cannot be improved in neither of the algorithms; additionally, the dependency on $d$ is almost optimal.
We refer the reader to \Cref{sec:lowerbounds} 
for proper formal discussion.

\subsection{Related Work}
\label{sec:related}

\probECCshort\ and \probECPshort\ are classical combinatorial problems~\cite{Erdos1966representation,Hall1941partitions,Lovasz1968covering,Ryser1973intersection}, whose origins can be traced back to fundamental questions posed by Boole in 1868~\cite{Boole1952propositions}.  Over time, these problems have been studied under various names, including \textsc{Covering by Cliques} (Problem~GT17) and \textsc{Intersection Graph Basis} (Problem~GT59) in Garey and Johnson's compendium~\cite{garey-johnson}, as well as \textsc{Keyword Conflict}~\cite{kellerman}. 

From a practical standpoint,   \probECCshort\ and \probECPshort\ arise in diverse application areas such as computational geometry~\cite{cc-apl1}, applied statistics~\cite{cc-apl2,cc-apl3}, networks \cite{guillaume:cc}, compiler optimization~\cite{cc-apl4}, or bioinformatics~\cite{figueroa2004clustering}.
Due to their broad importance, both \probECCshort\ and \probECPshort\ have been investigated from multiple perspectives, including approximation algorithms and inapproximability bounds~\cite{apx:cc,lund-yannakakis}, heuristics~\cite{bt:cc,cc-apl2,kellerman,kou:cc,cc-apl3,cc-apl4}, and polynomial-time solutions for special 
graph classes~\cite{cerioli2008partition,garey-johnson,hoover1992complexity,cc-class2,hunt1998complexity,cc-class1,orlin:cc}.

    Moreover, the relation between the smallest edge clique cover and the maximum independent set plays another important role in practical applications. On large instances, it may be infeasible to compute the optimal edge clique cover, and only heuristical methods could be applied within a reasonable timeframe. Finding a suffiently large independent set serves then as a certificate for the quality of the solution, in terms of the size of the edge clique cover: If, say, there is only a 1\% difference between the found independent set and clique cover, then the smallest clique cover is also bound to lie in this small interval. As observed by Hevia et al.~\cite{HeviaECCexperiments}, this phenomenon occurs frequently on real-world instances. Simultaneously, this certifies that the difference between the smallest edge clique cover and the largest independent set is also small on these instances.

Both \probECCshort\ and \probECPshort\ are known to be NP-complete even in  restricted graph classes~\cite{orlin:cc,chang2001tree,hoover1992complexity}.  In particular,  \probECCshort remains NP-complete and even when
 the input graph is planar \cite{chang2001tree} or has bounded degree \cite{hoover1992complexity}.
 Furthermore, Lund and Yannakakis~\cite{lund-yannakakis} showed that \probECCshort\ is not approximable within a factor of $n^\varepsilon$ (for some $\varepsilon > 0$) unless \classP=\classNP.

In particular, a significant amount of work in the algorithms engineering community is devoted to the studies of  \probECCshort  on sparse graphs, whose independence number is large. For example,  Abdullah and Hossain \cite{abdullah2022sparse}, and 
Conte,   Grossi,  and Marino \cite{conte2020large} studied \probECCshort  on $d$-degenerate graphs. Blanchette, Kim, and   Vetta in \cite{BlanchetteKV12} obtained an PTAS for  \probECCshort  on planar graphs as well as an FPT algorithm parameterized by the treewidth of a graph. 

Within the realm of parameterized complexity, a natural parameterization is by the number of cliques $k$. Gramm et al.~\cite{GrammGHN08} pioneered this direction by presenting simple reduction rules that produce a kernel of size at most $2^k$, making it one of the earliest examples of kernelization techniques described in textbooks~\cite{CyganFKLMPPS15,kernelizationbook19,Niedermeierbook06}. Later, Cygan, Pilipczuk, and Pilipczuk~\cite{CyganPP16} showed that there is no algorithm solving \probECCshort\ in time $2^{2^{o(k)}}\cdot n^{\Oh(1)}$ 
unless the Exponential Time Hypothesis (ETH) fails.

 Mujuni and Rosamond~\cite{mujuni2008parameterized} established that   \probECPshort is FPT when parameterized by $k$, the number of cliques in the solution.  Feldmann, Issac, and Rai~\cite{FeldmanECP20} improved the running time to $2^{k^{3/2}\log k}\cdot n^{\Oh(1)}$.  Fleischer and Wu~\cite{Fleischer2010edge} devised an algorithm for planar graphs running in $2^{\Oh(\sqrt{k})}\cdot n^{\Oh(1)}$, and for graphs of degeneracy $d$ they achieved $2^{dk} \cdot k \cdot n^{\Oh(1)}$ running time.

 \subsection{Organization of the Paper}
 In~\Cref{sec:defs}, we introduce basic notions and provide auxiliary results. \Cref{sec:partaboveIS} contains the proof of \Cref{thm:ecpis-fpt}. In \Cref{sec:ECC}, we show the dichotomy for \probECCISshort on general graphs, proving \Cref{thm:hard,cor:poly-zero-one}. Finally, in \Cref{sec:FPTbyomega}, we give algorithmic results and establish computational lower bounds for \probECCISshort on sparse graph classes, as stated in \Cref{cor:sparse}.
 
\section{Definitions and Preliminaries}\label{sec:defs}
\medskip\noindent\textbf{Graphs.}
In this paper, we consider simple undirected graphs and refer to the textbook by Diestel~\cite{Diestel} for notions that are not defined here.
 We use $V(G)$ and $E(G)$ to denote the set of vertices and the set of edges of $G$, respectively. We use~$n$ and~$m$ to denote the number of vertices and edges in~$G$, respectively, unless this creates confusion.
For a vertex subset~$X \subseteq V(G)$, we use~$G[X]$ to denote the subgraph of~$G$ induced by the vertices of~$X$ and~$G -X$ to denote~$G[V(G) \setminus X]$. 
For a vertex $v\in V(G)$, we write $N_G(v)=\{u\in V(G)\mid uv\in E(G)\}$ to denote the \emph{open neighborhood} of $v$, and we use $N_G[v]=N_G(v)\cup\{v\}$ for the \emph{closed neighborhood}. For $X\subseteq V(G)$, $N_G(X)=\big(\bigcup_{v\in X}N_G(v)\big)\setminus X$ and $N_G[X]=\bigcup_{v\in X}N_G[v]$. 
We denote by $d_G(v)=|N_G(v)|$ the \emph{degree} of a vertex $v$; a vertex of degree zero is \emph{isolated}. 
The minimum degree is denoted by $\delta(G)$. A graph $G$ is \emph{$d$-degenerate} for an integer $d\geq 0$ if $\delta(H)\leq d$ for any subgraph $H$ of $G$. 
By $\operatorname{dg}(G)$ we denote the degeneracy of $G$, i.e.\ minimum $d$ for which $G$ is $d$-degenerate.
A graph $H$ is a \emph{minor} of $G$ if the graph isomorphic to $H$ can be obtained from $G$ by vertex/edge deletions and edge contractions. A graph $G$ is \emph{$H$-minor-free} if $H$ is not a minor of $G$.   
A set of pairwise adjacent vertices is called a \emph{clique}, and a set of pairwise non-adjacent vertices is \emph{independent}.  
The maximum size of an independent set in $G$ is denoted by $\alpha(G)$, and the maximum size of a clique in $G$ is denoted by $\omega(G)$. 
A vertex $v\in V(G)$ is said to be \emph{simplicial} if $N_G[v]$ is a clique, and we say that a clique $K$ is \emph{simplicial} if $K$ equals $N_G[v]$ for a simplicial vertex $v$.
An edge $uv\in E(G)$ is \emph{covered} by a clique $K$ if $u,v\in K$. A set  $\mathcal{K}$ of cliques is a \emph{vertex clique cover} if for every $v\in V(G)$, there is $K\in\mathcal{K}$ such that $v\in K$. We say that $\mathcal{K}$ is an \emph{edge clique cover} if for every $uv\in E(G)$, there is $K\in\mathcal{K}$ with $u,v\in K$, and $\mathcal{K}$ is an \emph{edge clique partition} if $\mathcal{K}$ is an edge clique cover such that any two cliques in $\mathcal{K}$ has at most one common vertex.
We use $\ecc(G)$ and $\ecp(G)$ to denote the minimum size of an edge clique cover and partition of $G$, respectively.

\medskip\noindent\textbf{Parameterized Complexity.} We refer to the textbook by Cygan et al.~\cite{CyganFKLMPPS15}  for an introduction to the area.
We remind that  a \emph{parameterized problem} is a language $L\subseteq\Sigma^*\times\mathbb{N}$  where $\Sigma^*$ is a set of strings over a finite alphabet $\Sigma$. An input of a parameterized problem is a pair $(x,k)$ where $x$ is a string over $\Sigma$ and $k\in \mathbb{N}$ is a \emph{parameter}. 
A parameterized problem is \emph{fixed-parameter tractable} (or \classFPT) if it can be solved in  $f(k)\cdot |x|^{\mathcal{O}(1)}$ time for some computable function~$f$.  
The complexity class \classFPT contains all fixed-parameter tractable parameterized problems. We use the \emph{Exponential Time Hypothesis} (ETH) of Impagliazzo, Paturi, and Zane   \cite{ImpagliazzoPZ01} to obtain conditional computational lower bounds.  Under ETH,  there is a positive real $\delta$ such that \textsc{$3$-Satisfiability} (\textsc{$3$-SAT}) with $n$ variables and $m$ clauses cannot be solved in $2^{\delta n}\cdot (n+m)^{\Oh(1)}$ time; in particular, \textsc{$3$-SAT} cannon be solved in time subexponential in $n$.

\medskip
We need the following auxiliary results about edge clique covers and independent sets. First, we observe that $\alpha(G)$ gives a lower bound for $\ecc(G)$ and $\ecp(G)$.

\begin{observation}\label{obs:lb}
For a graph $G$ without isolated vertices $\alpha(G)\leq \ecc(G)\leq \ecp(G)$.
\end{observation}

\begin{proof}
Consider an edge clique cover of $G$ of size $\ecc(G)$. Because $G$ has no isolated vertices, each vertex of any independent set is in a clique of $\mathcal{C}$. Since an independent set cannot contain two vertices in the same clique, we obtain that $\alpha(G)\leq\ecc(G)$. The second inequality is trivial because each edge clique partition is an edge clique cover.  
\end{proof}

Notice that the absence of isolated vertices is crucial for the lower bound of $\ecc(G)$ and $\ecp(G)$. We also remark that  that excluding isolated vertices is essential for our algorithmic results for \probECPISshort and \probECPISshort. 

\begin{observation}\label{obs:isolhard} 
For any graph class $\mathcal{G}$ closed under adding pendent neighbors and isolated vertices such that \probECPshort (\probECCshort) is \classNP-complete on $\mathcal{G}$, 
it is \classNP-complete to decide whether a graph $G\in\mathcal{G}$ given together with its maximum independent set has an edge clique partition (cover) of size at most $\alpha(G)$. 
\end{observation}

\begin{proof}
Consider a graph $G$. We construct the graph $G'$ by adding a pendent neighbor to each vertex of $G$ and adding $k$ isolated vertices.  Notice that $\alpha(G')=|V(G)|+k$ and the set of added pendent and isolated vertices is a maximum independent set. It remains to observe that because each edge incident to a degree one vertex should be covered by a separate clique, $G'$ has an edge clique partition (cover, respectively) of size at most $\alpha(G')$  if and only if $G$ has an edge clique partition (cover, respectively) of size at most $k$.
\end{proof}

The following lemma plays the most fundamental role for \probECCISshort and \probECPISshort.

\begin{lemma}\label{lem:nonsimpl}
Let $G$ be a graph without isolated vertices, and let $\mathcal{C}$ be an edge clique cover (partition) of $G$ of size at most $\alpha(G)+k$ for an integer $k\geq 0$. Then 
\begin{itemize}
\item[(i)] at most $k$ vertices of any maximum independent set are non-simplicial,
\item[(ii)] at most $2k$ cliques of $\mathcal{C}$ are non-simplicial.  
\end{itemize}
\end{lemma} 

\begin{proof}
As an edge clique partition is an edge clique cover, it is sufficient to show the claim when $\mathcal{C}$ is an edge clique cover of size at most $\alpha(G)+k$.
Let $X$ be a maximum independent set of $G$. Clearly, any clique of $G$ contains at most one vertex of $X$ by independence. 
Because $G$ has no isolated vertices, each vertex $v\in X$ is included in at least one clique of $\mathcal{C}$. Furthermore, $v$ is simplicial if $v$ is included in a single clique of $\mathcal{C}$. 
Since $|\mathcal{C}|\leq \alpha(G)+k$, we obtain that at least $\alpha(G)-k$ vertices of $X$ are contained in exactly one clique of $\mathcal{C}$.
Therefore, at least $\alpha(G)-k$ vertices of $X$ are simplicial, and at least $\alpha(G)-k$ cliques of $\mathcal{C}$ are simplicial. This means that at most $k$ vertices of $X$ are non-simplicial and proves (i). To show (ii), note that because 
 at least $\alpha(G)-k$ cliques of $\mathcal{C}$ are simplicial, at most $2k$ cliques could be non-simplicial. This concludes the proof. 
\end{proof}

In particular, this lemma implies that \probECCISshort and \probECPISshort are trivial for $k=0$. 

\begin{observation}\label{obs:kzero}
\probECCISshort and \probECPISshort can be solved in polynomial time for $k=0$.
\end{observation}

\begin{proof}
By \Cref{lem:nonsimpl}, any clique cover or partition of size at most $\alpha(G)$ should contain only simplicial cliques.  We list all pairwise distinct simplicial cliques using the algorithm 
of of Kloks, Kratsch, and M{\"{u}}ller~\cite{KloksKM00}, 
and check whether they compose an edge clique cover or partition, respectively.  
\end{proof}

We use~\Cref{lem:nonsimpl} to argue that for graphs admitting an edge clique cover of bounded size, we can compute $\alpha(G)$. 

\begin{lemma}\label{lem:iscover}
Let $G$ be a graph and let $\mathcal{C}$ be an edge clique cover (partition) of $G$ of size at most $k$. Then $\alpha(G)$ can be computed in $\Oh(2^{k}\cdot k^2n)$ time.
\end{lemma}

\begin{proof}
It is sufficient to show the claim when $\mathcal{C}$ is an edge clique cover. We can assume without loss of generality that $G$ is a graph with at least one edge that  has no isolated vertices because isolated vertices are included in any maximum independent set. Thus, it is sufficient to compute a maximum independent set in the graph obtained from $G$ by deleting isolated vertices.  Notice that $\alpha(G)\leq k$ as each vertex of an independent set is included in a clique of $\mathcal{C}$, and each clique contains at most one vertex of the set. 

Let $\mathcal{C}=\{C_1,\ldots,C_k\}$. For every vertex $v\in V(G)$, we consider its \emph{characteristic vector} $\bfx_v=(x_1,\ldots,x_k)$ where 
\[x_i=
\begin{cases}
1,\mbox{if }v\in C_i\\
0,\mbox{otherwise}
\end{cases}
\]
for $i\in\{1,\ldots,k\}$.
Observe that two vertices $u$ and $v$ are nonadjacent if and only if $\bfx_u$ and $\bfx_v$ are orthogonal, that is, $\bfx_u\cdot \bfx_v=0$. 
For a set $U\subseteq V(G)$, we write 
$\bfx_U=\sum_{v\in U}\bfx_v$, and we assume that $\bfx_\emptyset=\mathbf{0}$. 
 Notice that if $U$ is an independent set, then, by the orthogonality of characteristic vectors, $\bfx_U$ is a 0--1 vector. 
For two 0--1 vectors $\bfx=(x_1,\ldots,x_k)$ and $\bfy=(y_1,\ldots,y_k)$, we write $\bfx\preceq \bfy$ if $x_i\leq y_i$ for all $i\in\{1,\ldots,k\}$.

We use dynamic programming to compute $\alpha(G)$. For this, for each 0--1 vector $\bfx\in \{0,1\}^k$ and each $i\in\{0,\ldots,k\}$, we 
we compute the values of  the Boolean function 
\begin{equation*}
\alpha(\bfx,i)=
\begin{cases}
{\sf true },\mbox{if there is an independent set }U\text{ of size }i\text{ s.t. }\bfx_U=\bfx\\
{\sf false},\mbox{otherwise} 
\end{cases}
\end{equation*}   
We initialize computation, by setting 
\begin{equation}\label{eq:init}
\alpha(\bfx,0)=
\begin{cases}
{\sf true},\mbox{if }\bfx=\mathbf{0}\\
{\sf false},\mbox{otherwise }
\end{cases}
\end{equation}
For $i\geq 1$, we use the recursion
\begin{equation}\label{eq:rec}
\alpha(\bfx,i)=\bigvee_{v\in V(G)\text{ s.t.\ }\bfx_v\preceq \bfx}\alpha(\bfx-\bfx_v,i-1)
\end{equation}
Then $\alpha(G)$ is the maximum value of $i\in\{1,\dots,k\}$ such that there is a 0--1 vector $\bfx$ with $\alpha(\bfx,i)=\mathsf{true}$. 

The correctness of the computation of the values of $\alpha(\bfx,i)$ in \Cref{eq:init} and \Cref{eq:rec} follows from the orthogonality of the characteristic vectors for the vertices of an independent set by standard dynamic programming arguments. To evaluate the running time, notice that we compute $\alpha(\bfx,i)$ for at most $2^k$ vectors $\bfx$ and at most $k+1$ values of $i$. Then, for given $\bfx$ and $i$, we need $\Oh(kn)$ time to compute $\alpha(\bfx,i)$. Thus, the overall running time is $\Oh(2^{k}\cdot k^2n)$. 
This completes the proof.
\end{proof}

We remark that our algorithm from~\Cref{lem:iscover} can be easily modified to output an independent set of maximum size.

\section{Clique Partition above Independent Set}\label{sec:partaboveIS}

In this section, we prove  \Cref{thm:ecpis-fpt}. For this, we need several auxiliary results. 
We start with the classical result of 
de Bruijn and  Erd\H{o}s from 1948~\cite{deBruijnE48}.

\begin{proposition}[de Brujin and Erd\H{o}s \cite{deBruijnE48}]\label{prop:dBE}
For any $n\geq 3$, an edge clique partition of $K_n$ into cliques of size at most $n-1$ contains at least $n$ cliques. 
\end{proposition}

To solve \probECPISshort, we use as a black box an algorithm for \probECPshort parameterized by the number of cliques in a solution.  The state-of-the-art \classFPT algorithm was given by 
Feldmann,  Issac, and Rai~\cite{FeldmanECP20}.
 
\begin{proposition}[Feldmann, Issac, Rai \cite{FeldmanECP20}]\label{prop:ecp}
 \probECPshort can be solved in $2^{\Oh(k^{3/2}\log k)}\cdot n^{\Oh(1)}$ time where $k$ is the maximum number of cliques in a solution. Furthermore, the algorithm outputs an edge clique partition of size at most $k$ if it exists. 
 \end{proposition}

  Next, we show that combining \Cref{lem:nonsimpl}, \Cref{lem:iscover}, and \Cref{prop:ecp}, we can compute $\alpha(G)$ for graphs admitting an edge clique partition of size at most $\alpha(G)+k$. 
  
 \begin{lemma}\label{lem:isecp}
 There is an algorithm with running time $2^{\Oh(k^{3/2}\log k)}\cdot n^{\Oh(1)}$ that, given an instance  $(G,k)$ of \probECPISshort,
either computes $\alpha(G)$ or correctly reports that $(G,k)$ is a no-instance of  \probECPISshort.
  \end{lemma} 
  
 \begin{proof}
 Let $G$ be a graph without isolated vertices. We compute the set $\mathcal{S}$ of all simplicial cliques and select a simplicial vertex in each simplicial clique. Denote by $X$ the obtained set of simplicial vertices. It is well-known that $G$ has a maximum independent set $U$ such that $X\subseteq U$. This implies that
 $\alpha(G)=\alpha(H)+|\mathcal{S}|$ where $H=G[W]$ for $W=V(G)\setminus\bigcup_{S\in\mathcal{S}}S$. If $G$ admits an edge clique partition of size at most $\alpha(G)+k$ then, by \Cref{lem:nonsimpl}, at most $2k$ cliques of the partition are non-simplicial. Notice that the edges of $H$ can be covered only by non-simplicial cliques. Therefore, $H$ should have an edge clique partition of size at most $2k$. We apply the algorithm from \Cref{prop:ecp} to $H$ and check in $2^{\Oh(k^{3/2}\log k)}\cdot n^{\Oh(1)}$ time whether $H$ has an edge clique partition of size at most $k'=2k$. If the algorithm reports that such a partition does not exist then we conclude that $(G,k)$ is a no-instance of \probECPISshort.  Otherwise, the algorithm outputs an edge clique partition $\mathcal{C}$ of size at most $2k$. Given this partition, we use the algorithm from   \Cref{lem:iscover} to compute $\alpha(H)$ in $2^{\Oh(k)}\cdot n^{\Oh(1)}$ time. Then we output $\alpha(G)=\alpha(H)+|\mathcal{S}|$. Because simplicial cliques can be found in polynomial time, for example, by the algorithm of Kloks, Kratsch, and M{\"{u}}ller~\cite{KloksKM00}, the overall running time is  $2^{\Oh(k^{3/2}\log k)}\cdot n^{\Oh(1)}$. This concludes the proof.
 \end{proof} 
  
By the following claim, we show that all cliques of size at least $6k+1$ that belong to any solution to an instance of \probECPISshort can be found in polynomial time.
 
 \begin{lemma}\label{lem:bigcliques}
There is a polynomial algorithm that, given an instance $(G,k)$ of \probECPISshort for $k\geq 1$, 
either outputs a set $\mathcal{K}$ of cliques of $G$ such that
\begin{itemize}
\item[(i)] each clique $C\in \mathcal{K}$ is included in any edge clique partition of $G$ of size at most $\alpha(G)+k$,
\item[(ii)] for every edge clique partition $\mathcal{C}$ of $G$ of size at most $\alpha(G)+k$ and any clique $C\in\mathcal{C}$ of size at least $6k+1$, $C\in \mathcal{K}$, 
\item[(iii)] for any two distinct cliques $C_1,C_2\in\mathcal{K}$, $|C_1\cap C_2|\leq 1$,
\end{itemize}
or correctly concludes that $(G,k)$ is a no-instance.
\end{lemma}

\begin{proof}
To prove the lemma, we use critical cliques introduced by Lin, Jiang, and
Kearney~\cite{LinJK00}. A clique $C$ of a graph $G$ is said to be \emph{critical} if $C$ is an inclusion-wise maximal clique consisting of true twins, that is, $N_G[u]=N_G[v]$ for any $u,v\in C$. Note that each vertex belongs to a unique critical clique. It is known (see~\cite{LinJK00}) that all critical cliques can be found in linear time.
We extensively use the fact that for any critical clique $K$ in $G$, $K\cup \{x\}$ is a clique in $G$ for every $x\in N_G(K)$.

Let $G$ be a graph without isolated vertices and $k\geq 1$ be an integer. We initialize $\mathcal{K}:=\emptyset$. We consider cliques of two types---cliques containing simplicial vertices and cliques without simplicial vertices.

Let $v$ be a simplicial vertex of $G$. We consider the auxiliary graph $H_v$ with $V(H_v)=N_G[v]$ such that two vertices $x,y\in V(H_v)$ are adjacent if and only if there is no simplicial clique $C\neq N_G[v]$ with $x,y\in C$. Then we find all critical cliques of $H_v$ using the result of ~\cite{LinJK00}. If the critical clique $K$ containing $v$ is of size at least $2k+1$ then we set $C=N_G[v]$ 
and $\mathcal{K}:=\mathcal{K}\cup\{C\}$ unless $\mathcal{K}$ already contains $C$. Otherwise, for each critical clique $C$ of size at least $2k$ that does not contain $v$, we define $C=N_{H_v}[K]$ and set $\mathcal{K}:=\mathcal{K}\cup\{C\}$ unless $\mathcal{K}$ already contains $C$.
To argue that (i) and (ii) are fulfilled for the cliques containing $v$, we prove the following claim.

\begin{claim}\label{cl:simplcorr}
For any edge clique partition $\mathcal{C}$ of $G$ with $|\mathcal{C}|\leq \alpha(G)+k$ and any clique $C$ included in $\mathcal{K}$ for $v$, $C\in\mathcal{C}$. Furthermore, for any clique $C\in \mathcal{C}$ of size at least $6k+1$ containing $v$, $C\in\mathcal{K}$. 		
\end{claim}

\begin{claimproof}
Let $C$ be a clique added to $\mathcal{K}$ for $v$. 

Suppose first that $C=N_G[v]$ is the simplicial clique containing $v$ constructed for the critical clique $K$ containing $v$ is of size at least $6k+1$.
By \Cref{prop:dBE}, we have that 
there is a clique $C'\in\mathcal{C}$ such that $K\subseteq C'$. Otherwise, the edges of $G[K]$ should be covered by at least $2k+1$  non-simplicial cliques contradicting \Cref{lem:nonsimpl}. By the same arguments, for any $x\in N_{H_v}(K)$, the clique $K\cup \{x\}$ is also a subset of some clique of $\mathcal{C}$.
Because $k\geq 1$ and $|K|>1$, and no two cliques in $\mathcal{C}$ can share two or more vertices, $C'$ is the only clique containing $K$, as well as each $x\in N_{H_v}(K)$.
We have that $N_{H_v}[v]=N_{G}[v]$ since $v$ is simplicial and edges incident to $v$ cannot be contained in any simplicial clique other than $C=N_{G}[v]$.
We have that $N_{H_v}[K]=N_{H_v}[v]$ since $K$ is a critical clique.
We obtain that $C=N_G[v]=N_{H_v}[K]\subseteq C'$. To show that $C'=C$, notice that  $C$ is an inclusion maximal clique of $G$ (as well as any other simplicial clique).
Thus, we have that  $C\in\mathcal{C}$.  

Assume that $C=N_{H_v}[K]$ for a critical clique $K$ of size at least $2k$ such that $v\notin K$. Consider the clique $K'=K\cup\{v\}$ which has size at least $2k+1$.
Since edges of $K'$ belong to the only simplicial clique in $G$,  \Cref{prop:dBE} and \Cref{lem:nonsimpl} imply that there is $C'\in\mathcal{C}$ such that $K'\subseteq C'$ (otherwise we would have to cover edges of $K'$ with at least $2k+1$ non-simplicial cliques).
Since $|K|>1$, $C'$ is the only clique in $\mathcal{C}$ that contains $K$ and $K'$. 
Consequently, $C'$ should contain $K\cup\{x\}$ for each $x\in N_{H_v}(K)$, so $C=N_{H_v}[K]\subseteq C'$. We claim that $C'=C$. 
For the sake of contradiction, assume that there is $x\in C'\setminus C$. 
By the definition of $H_v$, we have that $x$ is a non-simplicial vertex of $G$, and for each $y\in K$, $xy\notin E(H_v)$ and there is a simplicial clique $K_y$ such that $x,y\in K_y$ with a simplicial vertex $z_y \in K_y\setminus\{x,y\}$, but $z_y\notin C'$.
Since $yy'\in E(H_v)$ for $y,y'\in K$, we have that $y,y'$ belong to distinct simplicial cliques, so $z_y\neq z_{y'}$ for distinct $y,y'\in K$, and $z_y$ is not adjacent with $z_{y'}$. 
Since $x\in C'$, we have that $K_y\notin \mathcal{C}$ for each $y\in K$. 
Then for each $y\in K$, $xz_y$ and $yz_y$ are covered by two distinct non-simplicial cliques.
This means that $\mathcal{C}$ contains at least $4k>2k$ non-simplicial cliques contradicting \Cref{lem:nonsimpl}.  This proves that $C'=C$ and completes the proof that 
 $C\in\mathcal{C}$ for each clique $C$ included in $\mathcal{K}$ for $v$. 

To prove the second part of the claim, assume that there is $C\in \mathcal{C}$ of size at least $6k+1$ such that $v\in C$. 
Consider a simplicial  vertex $u\notin V(H_v)$ such that $|N_G(u)\cap C|\geq 2$. Then the edges $ux$ for $x\in N_G(u)\cap C$ are covered by pairwise distinct non-simplicial cliques in $\mathcal{C}$. Also, if $u,u'\notin V(H_v)$ are nonadjacent simplicial vertices  then the edges $ux$ and $u'y$ for $x,y\in C$ are covered by distinct cliques. This implies that 
$X=\bigcup (N_G(u)\cap C)$, where the union is taken over all simplicial vertices $u\notin V(H_v)$ such that $|N_G(u)\cap C|\geq 2$, is of size at most $2k$ by \Cref{lem:nonsimpl}. 
Let $Y=C\setminus X$. We have that $Y$ is a clique of $H_v$ of size at least $4k+1$ containing $v$. Consider the edges $xy\in E(H_v)$ such that $y\in Y$ and $x\in V(H_v)\setminus Y$. Because these edges are not covered by simplicial cliques and one endpoint of each edge is in the clique $Y\subset C\in\mathcal{C}$, these edges are covered by pairwise distinct non-simplicial cliques in $\mathcal{C}$. By \Cref{lem:nonsimpl}, we obtain that $Z=\{y\in Y\mid N_{H_v}(y)\setminus Y\neq \emptyset\}$ is of size at most $2k$.
Let $K=Y\setminus Z$. We have that $K$ is 
a clique of size at least $2k+1$ containing true twins in $H_v$. Then there is a critical clique $K'\supseteq K$. 
In particular, $K'$ is either a critical clique of size at least $2k+1$ containing $v$ or a critical clique of size at least $2k$ that does not contain $v$. In both cases, our algorithm includes $C'=N_{H_v}[K']$ in $\mathcal{K}$. 
By the first part of the claim, $C'\in \mathcal{C}$.
Since $|C\cap C'|\ge |K|>1$, $C$ and $C'$ coincide.
This concludes the proof. 
\end{claimproof}

We run the described procedure for all simplicial vertices of $G$. By \Cref{cl:simplcorr}, the constructed $\mathcal{K}$ contains cliques with simplicial vertices satisfying (i) and (ii). In the next step, we deal with cliques without simplicial vertices.

We construct the auxiliary graph $H$ with $V(H)=V(G)$ such that two distinct vertices $x,y\in V(H)$ are adjacent if and only if there is no simplicial clique $C$ with $x,y\in C$. Then we find all critical cliques of $H$ using the algorithm from~\cite{LinJK00}. For each critical clique $K$ of size at least $2k+1$, we set $C=N_H[K]$. If $C$ is not a clique of $G$ then we stop and return that $(G,k)$ is a no-instance. Otherwise, we set $\mathcal{K}:=\mathcal{K}\cup\{C\}$ unless $C$ is already in $\mathcal{K}$. 
To argue correctness, we show  the following claims.

\begin{claim}\label{cl:nonsimplcorr}
For any edge clique partition $\mathcal{C}$ of $G$ with $|\mathcal{C}|\leq \alpha(G)+k$ and any clique $C$ without simplicial vertices included in $\mathcal{K}$, $C\in\mathcal{C}$. Furthermore, the algorithm correctly reports a no-instance. 
\end{claim}

\begin{claimproof}
Let $K$ be a critical clique of $H$ of size at least $2k+1$. By \Cref{prop:dBE} and \Cref{lem:nonsimpl}, there is $C'\in \mathcal{C}$ such that $K\subseteq C'$ as the edges of $H$ cannot be covered by simplicial cliques. Furthermore, for any  $x\in N_H(K)$, $K\cup\{x\}\subseteq C'$. Thus, if there are two vertices $x,y\in N_H(K)$ that are not adjacent in $G$, then an edge clique partition of size at most $\alpha(G)+k$ does not exist. In particular, this means that  the algorithm correctly reports a no-instance. Otherwise, if $C=N_H[K]$ is a clique of $G$, we have that there is $C'\in \mathcal{C}$ such that $C=N_H[K]\subseteq C'$. 

We claim that $C'=C$. The proof is by contradiction. Assume that there is $x\in C'\setminus C$. For each $y\in K$, $xy\notin E(H)$. Thus, for each $xy$ of this type, $x$ and $y$ are in a same simplicial clique $K_y$. Because for any edge $yy'\in E(H)$, $y$ and $y'$ cannot belong to a same simplicial clique in $G$, $K_y$ and $K_{y'}$ are distinct simplicial cliques for distinct $y,y'\in K$.
For every $y\in K$, $K_y$ contains a simplicial vertex $z_y$ distinct from $x,y$, and $z_y z_{y'}\notin E(G)$ for distinct $y,y'\in K$.
Edges of form $yz_y$ for $y\in K$ cannot belong to a same clique in $G$, and, since $K_y\notin \mathcal{C}$, they should be covered by non-simplicial cliques in $\mathcal{C}$.
We obtain that $\mathcal{C}$ contains at least $2k+1$ non-simplicial cliques contradicting \Cref{lem:nonsimpl}. This proves that $C'=C$ and concludes the proof of the claim.
\end{claimproof}
 
\begin{claim}\label{cl:nonsimplcorrbig}
For any edge clique partition $\mathcal{C}$ of $G$ with $|\mathcal{C}|\leq \alpha(G)+k$ and any clique $C\in \mathcal{C}$ of size at least $6k+1$ that does not contain a simplicial vertex, $C\in\mathcal{K}$. 	
\end{claim}

\begin{claimproof}
Let $C\in\mathcal{C}$ be a clique of size at least $6k+1$ that does not contain any simplicial vertex. 
Consider a simlicial  vertex $v$ such that $|N_G(v)\cap C|\geq 2$. Then the edges $vx$ for $x\in N_G(v)\cap C$ are covered by pairwise distinct non-simplicial cliques in $\mathcal{C}$. Also, if $v$ and $v'$ are nonadjacent simplicial vertices  then the edges $vx$ and $v'y$ for $x,y\in C$ are covered by distinct cliques. Therefore,
$X=\bigcup (N_G(v)\cap C)$, where the union is taken over all simplicial vertices $v$ such that $|N_G(v)\cap C|\geq 2$, is of size at most $2k$ by \Cref{lem:nonsimpl}.
Let $Y=C\setminus X$. We have that $Y$ is a clique of $H$ of size at least $4k+1$. Consider edges $xy\in E(H)$ such that $y\in Y$ and $x\in V(H)\setminus Y$. Because these edges are not covered by simplicial cliques and one endpoint of each edge is in the clique $Y$, these edges are covered by distinct non-simplicial cliques in $\mathcal{C}$. By \Cref{lem:nonsimpl}, we obtain that $Z=\{y\in Y\mid N_{H}(y)\setminus Y\neq \emptyset\}$ is of size at most $2k$. Let $K=Y\setminus Z$. We have that $K$ is a clique of $H$ of size at least $2k+1$ consisting of true twins of $H$. Therefore, there is a critical clique $K'\supseteq K$ of $H$. 
Then our algorithm should include $C= N_H[K']$ in $\mathcal{K}$. This concludes the proof. 
\end{claimproof}

Summarizing, we obtain that our algorithm either constructs a set $\mathcal{K}$ of cliques of $G$ satisfying (i) and (ii)
or correctly concludes that $(G,k)$ is a no-instance. To satisfy (iii), notice that if there are two distinct cliques $C_1,C_2\in\mathcal{K}$ with $|C_1\cap C_2|\geq 2$ then $(G,k)$ is a no-instance because of (i). Then we add a step to our algorithm where we check whether $|C_1\cap C_2|\geq 2$  for two cliques $C_1,C_2\in\mathcal{K}$, and report a no-instance if we find such pair of distinct cliques. This concludes the description of the algorithm and its correctness proof.

To evaluate the running time, note that all simplicial vertices can be listed in polynomial time using, for example, the algorithm of Kloks, Kratsch, and M{\"{u}}ller~\cite{KloksKM00}.
Then we can find simplicial cliques and construct the auxiliary graphs $H_v$ for simplicial vertices $v$ and $H$ in polynomial time. Taking into account that the construction of critical cliques can be done in linear time~\cite{LinJK00}, we obtain that the overall running time is polynomial. This concludes the proof.
\end{proof}

Now we are ready to show  \Cref{thm:ecpis-fpt} which we restate. 

\ECPISFPT*

\begin{proof}
Let $(G,k)$ be an instance of   \probECPISshort. As pointed out in~\Cref{obs:kzero}, the problem is trivial if $k=0$. Hence, we assume that $k\geq 1$. 
First, we call the algorithm from \Cref{lem:isecp} for $(G,k)$. The algorithm either computes $\alpha(G)$ or correctly reposts that $(G,k)$ is a no-instance of  \probECPISshort.
If we obtain a no-instance, we report it and stop. We assume from now on that this is not the case, and the value $\alpha(G)$ is known to us.

In the following step, we construct a partial solution $\mathcal{K}$ using the algorithm from \Cref{lem:bigcliques}.  This algorithm works in polynomial time and 
either outputs a set $\mathcal{K}$ of cliques of $G$ such that
\begin{itemize}
\item[(i)] each clique $C\in \mathcal{K}$ is included in any edge clique partition of $G$ of size at most $\alpha(G)+k$,
\item[(ii)] for every edge clique partition $\mathcal{C}$ of $G$ of size at most $\alpha(G)+k$ and any clique $C\in\mathcal{C}$ of size at least $6k+1$, $C\in \mathcal{K}$, 
\item[(iii)] for any two distinct cliques $C_1,C_2\in\mathcal{K}$, $|C_1\cap C_2|\leq 1$,
\end{itemize}
or correctly concludes that $(G,k)$ is a no-instance. If the algorithm returns that $(G,k)$ is a no-instance then we return the answer and stop. We also conclude that $(G,k)$ is a no-instance if $|\mathcal{K}|>\alpha(G)+k$. From now on, we assume that $\mathcal{K}$ is a set of  cliques that are included in any edge clique partition of size at most $\alpha(G)+k$ such that any two distinct cliques of $\mathcal{K}$ do not cover the same edge.

Next, we list all distinct simplicial cliques of $G$ and denote by $\mathcal{S}$ the set of simplicial cliques. We say that a clique $S\in\mathcal{S}$ is \emph{broken} in a (potential) solution to $(G,k)$, if $S$ is not included in the edge clique partition. Notice that if $S$ is a broken simplicial clique, then at least two distinct non-simplicial cliques should be used to cover the edges incident to the corresponding simplicial vertex. By \Cref{lem:nonsimpl}, any edge clique partition of size at most $\alpha(G)+k$ can include at most $2k$ non-simplicial cliques. Thus, for any solution, the number of broken simplicial cliques should not exceed $k$. Our algorithm identifies the set $\mathcal{B}$ of broken cliques in a solution. 

We initialize $\mathcal{B}$ using the following branching algorithm. Consider the auxiliary graph $H$ whose nodes are simplicial cliques. Two distinct nodes $S_1,S_2\in\mathcal{S}$ are adjacent in $H$ if and only if $|S_1\cap S_2|\geq 2$. Observe that if $|S_1\cap S_2|\geq 2$ then either $S_1$ or $S_2$ should be broken in a solution. Thus, the set of broken cliques in any solution should be a vertex cover of $H$ of size at most $k$. We can decide in $\Oh(2^k\cdot n)$ time whether $H$ has a vertex cover of size at most~$k$ and, moreover, we can list all inclusion minimal vertex covers using the standard recursive branching algorithm for \textsc{Vertex Cover} (see~\cite{CyganFKLMPPS15}). If $H$ has no vertex cover of size at most~$k$, we report that $(G,k)$ is a no-instance of  \probECPISshort and stop. Assume that this is not the case. Then we obtain the list $\mathcal{L}$ of all inclusion minimal vertex covers of $H$ of size at most $k$. Notice that $|\mathcal{L}|\leq 2^k$. Then if $(G,k)$ admits an edge clique partition of size at most $\alpha(G)+k$, there is $L\in \mathcal{L}$ such that 
$L\subseteq \mathcal{B}$ --- the set of broken cliques in $\mathcal{C}$.  To initialize $\mathcal{B}$, we branch on all possible choices of $L\in \mathcal{L}$ and set $\mathcal{B}:=L$.
If we find a solution for some choice of $L$, we output it and stop. Otherwise, if we fail to find an edge clique partition of size at most $\alpha(G)+k$ for any initial selection of $\mathcal{B}$, we report that $(G,k)$ is a no-instance of  \probECPISshort.

From now on, we assume that the initial choice of $\mathcal{B}$ is fixed from one of at most $2^k$ choices given by $\mathcal{L}$. 
We update $\mathcal{B}$ by adding cliques that get broken because of the cliques in the partial solution. Notice that if $S\notin \mathcal{K}$ is a simplicial clique covering the same edge as one of the cliques in the partial solution, then $S$ has to be broken. This implies the correctness of the following step: 
\begin{itemize}
\item For every simplicial clique $S\notin\mathcal{B}\cup\mathcal{K}$ such that $|S\cap K|\geq 2$ for some $K\in\mathcal{K}$, set $\mathcal{B}:=\mathcal{B}\cup \{S\}$. \
\item If $|\mathcal{B}|>k+1$ then stop and discard the current choice of $L$. 
\end{itemize}
Assume that the algorithm does not stop in this step. We call the set $\mathcal{B}$ obtained in this step a \emph{base set of broken cliques}.

We say that a simplicial clique $S$ is \emph{free} if $S\notin \mathcal{B}\cup\mathcal{K}$, and we use $\mathcal{F}$ to denote the set of free cliques.
\Cref{lem:bigcliques} 
and the construction of $\mathcal{B}$ imply the following  property.

\begin{claim}\label{cl:free}
For any two distinct cliques $S_1,S_2\in \mathcal{F}\cup\mathcal{K}$, $|S_1\cap S_2|\leq 1$.
\end{claim}
 
This means that no edge of $G$ is covered by two cliques of $\mathcal{F}\cup\mathcal{K}$. In the final step of our algorithm, we try to extend $\mathcal{F}\cup\mathcal{K}$ to an edge clique partition of $G$ of size at most $\alpha(G)+k$. If we fail, then we recurse by including one of the free cliques to the set of broken cliques. Formally, we use the subroutine 
\textsc{Extend$(\mathcal{F},\mathcal{B})$} in \Cref{alg:extend} which either finds a solution  extending $\mathcal{K}$, such that all cliques in $\mathcal{B}$ are broken,
or discards the current choice of $\mathcal{F}$ and $\mathcal{B}$.

\medskip
\begin{algorithm}[h]
\caption{\textsc{Extend$(\mathcal{F},\mathcal{B})$}} \label{alg:extend}
\KwIn{$\mathcal{F}$, $\mathcal{B}$}
\KwResult{ An edge clique cover of $G$ or {\sf No}}
\Begin
{ 
\lIf{$|\mathcal{B}|>k$}{\Return {\sf No} and quit}\label{ln:i}
construct the graph $Q$ with the edge set $R=\{uv\in E(G)\mid\text{there is no clique }K\in\mathcal{F}\cup\mathcal{K}\text{ s.t. }u,v\in K\}$,
whose vertex set is the set of the endpoints of the edges of $R$; note that $R$ is the set of edges which are not covered by $\mathcal{F}\cup\mathcal{K}$\label{ln:ii}\;
\lIf{$|V(Q)|>12k^2$}{\Return {\sf No} and quit}\label{ln:iii}
use the algorithm from \Cref{prop:ecp} to find an edge clique partition $\mathcal{C}$ of $Q$ of minimum size $k'\leq 2k$ if such a partition exists\label{ln:iv}\;
\If{the algorithm returns $\mathcal{C}$ \label{ln:va}}
{\lIf{$|\mathcal{C}^*|\leq\alpha(G)+k$ for $\mathcal{C}^*=\mathcal{C}\cup\mathcal{F}\cup\mathcal{K}$}
{\Return $\mathcal{C}^*$ and quit}\label{ln:v}}\label{ln:vb}
\ForEach{$uv\in E(G)$ such that $u,v\in V(Q)$ and there is $F\in\mathcal{F}$ with $u,v\in F$\label{ln:vi}}
{
	\If{a call to \textsc{Extend$(\mathcal{F}\setminus\{F\},\mathcal{B}\cup\{F\})$}\label{ln:via} returns a solution $\mathcal{C}^*$}
	{
		\Return{$\mathcal{C}^*$} and quit\label{ln:vii}\;
	}

}
{\Return {\sf No}}\label{ln:viia}
}
\end{algorithm}
\medskip

 The description of the algorithm is finished.
 To argue correctness, observe that in each recursive call of the algorithm, we increase the size of $\mathcal{B}$ and stop if $|\mathcal{B}|>k$ in line~\ref{ln:i}. This means that \textsc{Extend$(\mathcal{F},\mathcal{B})$} is finite. If this subroutine outputs  a set of cliques $\mathcal{C}^*$ then it is a valid edge partition of $G$.
 First, because of \Cref{cl:free}, no two cliques intersect.
 Second, the cover $\mathcal{C}$, constructed by the algorithm from  \Cref{prop:ecp} in line~\ref{ln:iv}, outputs an edge clique partition for the edges which are not covered by the cliques of $\mathcal{F}\cup\mathcal{K}$. 
Since we verify that $|\mathcal{C}^*|\leq\alpha(G)+k$ in line~\ref{ln:v}, if the algorithm outputs some $\mathcal{C}^*$ then it is a valid solution to $(G,k)$.
Thus, we have to show that if $(G,k)$ is a yes-instance of  \probECPISshort, then the algorithm necessarily finds a solution. For this, we show the following claim.

\begin{claim}\label{cl:correcp}
Suppose that $G$ has an edge clique partition of size at most $\alpha(G)+k$ such that all simplicial cliques in $\mathcal{B}$ are broken. Suppose also that $\mathcal{B}$ contains the base set of broken cliques, and let 
$\mathcal{F}$ be the set of simplicial cliques $S$ such that $S\notin \mathcal{B}\cup\mathcal{K}$.
Then \textsc{Extend$(\mathcal{F},\mathcal{B})$} outputs a solution to $(G,k)$. 
\end{claim}

\begin{claimproof}
The proof is by induction on the number of broken cliques in $\mathcal{F}$.
 
The base case is when $G$ has an edge clique partition $\mathcal{C}'$ of size at most $\alpha(G)+k$ such that no clique of $\mathcal{F}$ is broken with respect to $\mathcal{C}'$.
Then $\mathcal{F}\cup \mathcal{K}\subseteq \mathcal{C}'$, and $\mathcal{B}$ contains all simplicial cliques that are broken with respect to $\mathcal{C}'$.
By \Cref{lem:nonsimpl} $|\mathcal{B}|\leq k$, and the algorithm does not stop in line~\ref{ln:i}. 
The graph $Q$ constructed in line~\ref{ln:ii} is the graph wtith edge set containing all edges that are not covered by the cliques of $\mathcal{F}\cup \mathcal{K}$. 
By \Cref{lem:bigcliques}, $\mathcal{C}'\setminus \mathcal{K}$ can only contains cliques of size at most $6k$.
This means that the cliques covering the edges of $Q$ are cliques in $\mathcal{C}'\setminus(\mathcal{K}\cup \mathcal{F})$ of size at most $6k$. Because these cliques are not simplicial, there are at most $2k$ such cliques in the partition by \Cref{lem:nonsimpl}. Therefore, $|V(Q)|\leq 12k^2$ and the algorithm does not stop in line~\ref{ln:iii}. Then the algorithm  from \Cref{prop:ecp} should find an edge clique partition of $Q$ of minimum size $k'\leq 2k$ such that $\mathcal{C}^*=\mathcal{C}\cup\mathcal{F}\cup\mathcal{K}$ is an edge clique partition of $G$ of size at most $\alpha(G)+k$. This concludes the proof for the base case.

Suppose that for any edge clique partition $\mathcal{C}'$ of size at most $\alpha(G)+k$ where all cliques of $\mathcal{B}$ are broken, at least one clique of $\mathcal{F}$ also gets broken.
By the same arguments as in the base case, the algorithm does not stop in line~\ref{ln:i}.
Since every vertex of $Q$ belongs to a clique in $\mathcal{C}'\setminus (\mathcal{K}\cup \mathcal{F})$, the size of $Q$ is bounded similarly to the base case and the algorithm does not stop in line~\ref{ln:iii} as well.
However, the algorithm from  line~\ref{prop:ecp} would either fail to find any edge clique partition of $Q$ of size at most $2k$ or would output an edge clique partition $\mathcal{C}$ that would be rejected by the algorithm in line~\ref{ln:v}. 

Suppose that the algorithm has reached line~\ref{ln:vb}.
We argue that at least one of the recursive branches of the algorithm in line~\ref{ln:vi} leads to the solution $\mathcal{C}'$.
For the sake of contradiction, assume that for each $F\in \mathcal{F}\setminus \mathcal{C}'$, we have that $|F\cap V(Q)|\le 1$.
To give a lower bound for $|\mathcal{C}'|$, note that $\mathcal{K}\subset \mathcal{C}'$ by definition of $\mathcal{K}$.
Moreover, each $F_i\in\mathcal{F}$ contains a simplicial vertex $v_i$.
These vertices form an independent set of size $|\mathcal{F}|$, and they should belong to pairwise-distinct cliques in $\mathcal{C}'$.
Hence, there are at least $|\mathcal{F}|$ cliques in $\mathcal{C}'$, and each of these cliques is a subset to some $F_i\in\mathcal{F}$.
Consequently, $\mathcal{C}'\setminus\mathcal{K}$ contains at least $|\mathcal{F}|$ cliques that do not cover edges of $Q$.
Let $\mathcal{C}''\subset \mathcal{C}'$ be a set of cliques that cover at least one edge of $Q$.
Transform each $C\in\mathcal{C}''$ by replacing it with $C\cap V(Q)$.
Note that $C\cap V(Q)$ is a clique in $Q$, since no edge in $G[V(Q)]$ is covered by $\mathcal{F}$.
Hence, the transformation of $\mathcal{C}''$ yields an edge clique partitioning of $Q$.
Consequently, $|\mathcal{C}'|\ge |\mathcal{K}|+|\mathcal{F}|+|\mathcal{C}''|\ge |\mathcal{K}|+|\mathcal{F}|+\ecp(Q)$.
But then $\ecp(Q)\le 2k$ and the algorithm should've returned the solution with $|\mathcal{C}^*|\le |\mathcal{C}'|$ in line~\ref{ln:v}. 

The contradiction shows that, because $(G,k)$ is a yes-instance, there is $F\in\mathcal{F}$ with $u,v\in F$ for $uv\notin E(Q)$, and the algorithm would consider the branch for $uv$ in line~\ref{ln:vi}. By the inductive assumption, the algorithm finds a solution for this branch and returns it in line~\ref{ln:vii}. This concludes the proof.
\end{claimproof}  

To complete the correctness proof, observe that if $(G,k)$ is a yes-instance of \probECPISshort then there is $L\in \mathcal{L}$ such that every clique of $L$ is broken in a solution. Then for the base set of broken cliques $\mathcal{B}$  constructed for $L$, each clique of $\mathcal{B}$ should be broken. By \Cref{cl:correcp},  \textsc{Extend$(\mathcal{F},\mathcal{B})$} called for this $\mathcal{B}$ and the corresponding set of free cliques $\mathcal{F}$ should output an edge clique partition of size at most $\alpha(G)+k$. This completes the correctness proof. 

We conclude the proof with a claim explaining the running time bound.
\begin{claim}
	The algorithm runs in $2^{\Oh(k^{3/2}\log k)}\cdot n^{\Oh(1)}$ time.
\end{claim}
\begin{claimproof}
	To evaluate the running time, note that $\alpha(G)$ is computed in time $2^{\Oh(k^{3/2}\log k)}\cdot n^{\Oh(1)}$ by~\Cref{lem:isecp}.
	Observe  that the partial solution 
	$\mathcal{K}$ is constructed in polynomial time by \Cref{lem:bigcliques}, and  recall that the set of simplicial cliques $\mathcal{S}$ can be obtained in polynomial time~\cite{KloksKM00}. 
	Given $\mathcal{S}$, the auxiliary graph $H$ can be constructed in polynomial time.
	Then the set $\mathcal{L}$ of size at most $2^k$ is constructed in $\Oh(2^k\cdot n)$ time. For each $L\in \mathcal{L}$, we construct the base set $\mathcal{B}$ of broken cliques and the corresponding set $\mathcal{F}$ of free cliques in polynomial time. Then we call \textsc{Extend$(\mathcal{F},\mathcal{B})$}. Steps in lines~\ref{ln:i}--\ref{ln:iii}, \ref{ln:va}--\ref{ln:vb}, and \ref{ln:vii} are trivially polynomial. 
	Step in line~\ref{ln:iv} demands $2^{\Oh(k^{3/2}\log k)}\cdot k^{\Oh(1)}$ time by \Cref{prop:ecp}. Becase $|V(Q)|\leq 12k^2$ in line~\ref{ln:vi}, we have at most $\binom{12k^2}{2}$ recursive calls in \Cref{ln:via}. Summarizing, each call of \textsc{Extend$(\mathcal{F},\mathcal{B})$} without taking account the running time for the recursive calls, is done in $2^{\Oh(k^{3/2}\log k)}\cdot k^{\Oh(1)}$ time. 
	Because we branch for at most $\binom{12k^2}{2}$ possibilities in line~\ref{ln:via} and the depth of the search tree is upper bounded by $k$, the total running time of   \textsc{Extend$(\mathcal{F},\mathcal{B})$} for each base set $\mathcal{B}$ and the corresponding set of free cliques is $k^{\Oh(k)}\cdot 2^{\Oh(k^{3/2}\log k)}\cdot n^{\Oh(1)}$. Then the overall running time is $2^k\cdot k^{\Oh(k)}\cdot 2^{\Oh(k^{3/2}\log k)}\cdot n^{\Oh(1)}$ which could be written as $2^{\Oh(k^{3/2}\log k)}\cdot n^{\Oh(1)}$. This concludes the running time analysis.
\end{claimproof}

The proof is complete.
\end{proof}

In the conclusion of this section, we remark that the algorithm of \Cref{prop:ecp} is a bottleneck for the running time of our algorithm for \probECPISshort. More precisely, given an algorithm solving  \probECPshort in time $f(k)\cdot n^{\Oh(1)}$,  the running time of our algorithm is $k^{\Oh(k)}f(2k)\cdot n^{\Oh(1)}$.
Furthermore, consider the graph $G'$ obtained from a graph $G$ by adding a pendent neighbor to each vertex of $G$. Since $\ecp(G')=\alpha(G')+\ecp(G)$, a faster algorithm for \probECPISshort would improve the algorithm of \Cref{prop:ecp}.

\section{Clique Cover above Independent Set}\label{sec:ECC}
In this section, we establish the dichotomy result for \probECCISshort by showing  \Cref{thm:hard} and \Cref{cor:poly-zero-one}. For this, we need some auxiliary results about relations of \probECCIS with other problems which will be useful also  in  \Cref{sec:FPTbyomega}. Therefore, we put them in a separate subsection.
 
 \subsection{Reductions between \probECCISshort, \probAECCshort, and \textsc{$k$-Coloring}}\label{sec:reductiontoannotated}
Following  Orlin~\cite{orlin:cc}, we consider the generalization of \probECCshort where the aim is to cover a subset of edges.

\defparproblema{\probAECC(\probAECCshort)}{A graph $G$, an edge set $B\subseteq E(G)$, an integer $k$.}{$k$}{Find $k$ cliques in $G$, such that each edge in $B$ belongs to at least one clique.}

 Recall that 
 $\ecc(G)$ denotes the size of the smallest edge clique cover of $G$. Similarly, in the context of \probAECCshort, for a graph $G$ and a subset of its edges $B \subseteq E(G)$, we denote by $\aecc_B(G)$ the size of the smallest collection of cliques in $G$ that cover all edges in~$B$.

\begin{lemma}\label{lemma:ecc_to_annotated}
    \probECCISshort admits an \classFPT-Turing-reduction to \probAECCshort, where an instance $(G,k)$ is reduced to at most $2k$ instances $(G',R',k')$ such that $G'$ is an induced subgraph of $G$ and $k'\le 2k$.
	The reduction works in time $4^{k} \cdot n^{\Oh(1)}$.
\end{lemma}

\begin{proof}
    Let $S$ be the set of simplicial vertices in $G$. Let $S' \subseteq S$ be a maximal subset of simplicial vertices containing no two vertices that are true twins of each other.
    It is easy to see that the property of being true twins defines an equivalence relation on $S$, therefore $S'$ contains exactly one vertex of each equivalence class defined by this relation. We first observe that most of $\alpha(G)$ is attributed to $S'$.

    \begin{claim}\label{claim:ecc_simplicial_size}
        If $(G, k)$ is a yes-instance, then the size of $S'$ is at least $\alpha(G) - k$.
    \end{claim}
    \begin{claimproof}
        Let $\mathcal{C}$ be the smallest edge clique cover of $G$. By \Cref{lem:nonsimpl}, at least $\alpha(G) - k$ cliques in $\mathcal{C}$ are simplicial. Since each simplicial clique in $\mathcal{C}$ contains a distinct vertex of $S'$, the claim follows.
    \end{claimproof}

    Next, we describe the reduced instances. Let $G' = G - S$, the graph obtained from $G$ by removing all simplicial vertices. Let $F \subseteq V(G')$ be the vertices that were adjacent to simplicial vertices in $G$, and let $B = E(G') \setminus E(G'[F])$ be the edges in $G'$ with at least one endpoint not in $F$. We observe that the difference between the smallest edge clique cover of $G$ and the smallest annotated edge clique cover of $G'$ with respect to $B$ is exactly $|S'|$, and the same holds for the difference between $\alpha(G)$ and $\alpha(G' - F)$. 

    \begin{claim}\label{claim:ecc_reduced}
        It holds that $\ecc(G) = \aecc_B(G') + |S'|$, and $\alpha(G) = \alpha(G' - F) + |S'|$.
    \end{claim}
    \begin{claimproof}
        Let $\mathcal{C}$ be the smallest edge clique cover of $G$. As observed in~\Cref{claim:ecc_simplicial_size}, $\mathcal{C}$ contains a simplicial clique for each vertex of $S'$. The remaining cliques of $\mathcal{C}$ cover all edges in $B$, which are exactly the edges of $G$ not contained in simplicial cliques of $G$. Therefore, $\aecc_B(G') \leq \ecc(G) - |S'|$. In the other direction, the cliques that realize $\aecc_B(G')$ together with simplicial cliques of $G$, one per a vertex in $S'$, cover all edges of $G$. Thus, $\ecc(G) \leq \aecc_B(G') + |S'|$, which concludes the first part of the claim.

        For the second part of the statement, let $I$ be the maximum independent set in $G$. We can assume that $S' \subseteq I$, since for each vertex $s \in S'$, $I$ contains exactly one vertex in its neighborhood, which could be safely replaced by $s$. No vertices of $S'$ are therefore in $F$ or $S \setminus S'$, and $I \setminus S'$ is an independent set in $G - (S \cup F) = G' - F$, thus $\alpha(G' - F) \geq \alpha(G) - |S'|$.
        In the other direction, let $I'$ be a maximum independent set in $G' - F$. The set $I' \cup S'$ is an independent set in $G$, since vertices of $G' - F$ are not adjacent to any simplicial vertices in $G$. Therefore, $\alpha(G) \geq \alpha(G' - F) + |S'|$, which establishes the second part of the claim.
    \end{claimproof}

    We are now ready to describe the algorithm claimed by the lemma. First, it computes $S$ and $S'$ by listing all pairwise distinct simplicial cliques with the algorithm 
of of Kloks, Kratsch, and M{\"{u}}ller~\cite{KloksKM00}.
    The algorithm then computes $G'$, $B$, $F$ as described above.
    If $G'$ or $B$ are empty, the algorithm sets $k' = 0$.
    Otherwise, it constructs instances $\mathcal{I} = (G', B, k')$ of \probAECCshort for increasing values of $k'$ from $1$ to $2k$.
    If all of them are classified as no-instances by the \probAECCshort algorithm, the algorithm reports that the original instance $(G, k)$ is a no-instance as well.
    Let $k'$ be the smallest value such that $\mathcal{I} = (G', B, k')$ is reported as a yes-instance.
    Let $\mathcal{C}'$ be the collection of cliques in $G'$ of size $k'$, covering all edges of $B$, as returned by the \probAECCshort algorithm.
    Finally, the algorithm uses~\Cref{lem:iscover} to find the maximum independent set in $G' - F$, as $\mathcal{C}'$ induces an edge clique cover of $G' - F$; let $t = \alpha(G' - F)$.
    The algorithm then reports $(G, k)$ to be a yes-instance if and only if $k' \leq t + k$.
    The respective edge clique cover of $G$ is constructed as a union of $\mathcal{C}'$ and simplicial cliques in $G$, one for each vertex of $S'$.

    By construction, it holds that at most $2k$ instances of \probAECCshort are produced, with the graph $G'$ being an induced subgraph of $G$, and the parameter at most $2k$. Next, we show the correctness of the reduction.
    Let $\mathcal{C}$ be the smallest edge clique cover of $G$, and let $I$ be the maximum independent set in $G$.
    Without loss of generality, assume that $S' \subseteq I$, and that for every $s \in S'$, the clique on $N[s]$ is contained in $\mathcal{C}$, as observed in Claims~\ref{claim:ecc_simplicial_size} and \ref{claim:ecc_reduced}.
    By \Cref{claim:ecc_reduced}, $|\mathcal{C}| - |I| = \aecc_B(G') - \alpha(G' - F)$. First, if $(G, k)$ is a yes-instance of \probECCISshort, then $|\mathcal{C}| - |I| \leq k$ and $\aecc_B(G') - \alpha(G' - F) \leq k$. Moreover, by \Cref{claim:ecc_simplicial_size}, $|S'|$ is at least $|I| - k$.
    Therefore, by \Cref{claim:ecc_reduced}, $\aecc_B(G') = |\mathcal{C}| - |S'| \leq 2k$, and $\mathcal{A}$ will identify $(G, B, k')$ as a yes-instance for $k' = \aecc_B(G')$. The algorithm then computes $t = \alpha(G' - F)$, and, since by the above $k' - t \leq k$, correctly concludes that $(G, k)$ is a yes-instance.

    In the other direction, if $(G, k)$ is a no-instance of \probECCISshort, then $|\mathcal{C}| - |I| > k$ and therefore $\aecc_B(G') - \alpha(G' - F) > k$. If $\aecc_B(G') > 2k$, then the algorithm correctly reports a no-instance, since all constructed instances of \probAECCshort are no-instances. Otherwise, let $k' = \aecc_B(G') \leq 2k$, and the algorithm finds $t = \alpha(G' - F)$. Since $k' - t > k$, the algorithm correctly reports that $(G, k)$ is a no-instance.

    Finally, for the running time, observe that all steps of the reduction are done in polynomial time, except the computation of maximum independent set in $G' - F$ via \Cref{lem:iscover}. Since~\Cref{lem:iscover} is called with an edge clique cover of size at most $2k$, the total running time of the reduction is bounded by $4^k \cdot n^{\Oh(1)}$.
\end{proof}

The running time of~\Cref{lemma:ecc_to_annotated} is dominated by computing the maximum independent set in an induced subgraph, and the reduction can be performed more efficiently on certain graph classes, where instead of~\Cref{lem:iscover} one can compute the independent set directly with a faster algorithm.
To accomodate for this, we state the following corollary of~\Cref{lemma:ecc_to_annotated}.

\begin{corollary}\label{cor:coloring_to_annotated_faster}
	Let $\mathcal{G}$ be a hereditary graph class. \probECCISshort admits an \classFPT-reduction to \probAECCshort, where an instance $(G,k)$ with $G \in \mathcal{G}$ is reduced to at most $2k$ instances $(G',R',k')$ such that $G'$ is an induced subgraph of $G$ and $k'\le 2k$. If it can be determined in time $f(k) \cdot n^{\Oh(1)}$ whether a graph in $\mathcal{G}$ has an independent set of size $k$, then the reduction admits the same running time bound.
\end{corollary}
\begin{proof}
    Perform the same reduction as given by~\Cref{lemma:ecc_to_annotated}, except for the final step of the algorithm. There, to compute $\alpha(G' - F)$, instead of~\Cref{lem:iscover}, use the independent set algorithm for $\mathcal{G}$ with the input $(G' - F, t)$ for increasing values of $t$ from $1$ to $k$.
    Note that since $\alpha(G) = \alpha(G' - F) + |S'|$ and $|S'| \geq \alpha(G) - k$ by Claims~\ref{claim:ecc_simplicial_size} and \ref{claim:ecc_reduced}, it holds that $\alpha(G' - F) \leq k$.
\end{proof}

Next, we show that \probAECCshort can be, in turn, reduced to the \textsc{$k$-Coloring} problem.
We note that the reduction is almost the same as the standard conversion from \probECCshort to \textsc{Vertex Clique Cover}, as introduced by Kou, Stockmeyer and Wong~\cite{kou:cc}.

\begin{lemma}\label{lemma:annotated_to_coloring}
	\probAECCshort is polynomial-time reducible to \textsc{$k$-Coloring} on a graph with $|B|$ vertices.
\end{lemma}
\begin{proof}
    Let $(G, B, k)$ be an instance of \probAECCshort. We construct the target graph $H$ of the \textsc{$k$-Coloring} instance as follows. We set $V(H) = B$, associating the vertices of $H$ with the edges of $B$. For the edge set of $H$, $ee' \in E(H)$ whenever two distinct edges $e, e' \in B$ are such that
    \begin{enumerate}
        \item either $e$ and $e'$ share a common endpoint, and the remaining endpoints are not adjacent in $G$,
        \item or $e$ and $e'$ do not share endpoints, and at least one edge between an endpoint of $e$ and an endpoint of $e'$ is missing in $G$.
    \end{enumerate}
    In other words, $e$ and $e'$ are adjacent in $H$ whenever they cannot be part of the same clique in $G$.

    Clearly, the construction is done in polynomial time. We now show the correctness of the reduction. First, assume there is a collection $\mathcal{C}$ of at most $k$ cliques in $G$ that covers all edges of $B$. We claim that each clique $C \in \mathcal{C}$ induces an independent set in $H$. Indeed, if edges $e$ and $e'$ are part of $C$, they are not adjacent in $H$, as both 1. and 2. in the definition of $E(H)$ contradict that $C$ is a clique.
    Consider now the cliques in $C$ as subsets of $V(H)$.
    If two sets contain the same element, remove it from one of them arbitrarily; repeat until no such sets occur.
    The resulting collection of subsets of $V(H)$ is a partition of $V(H)$, and each subset is an independent set in $H$, which shows that $H$ admits a $k$-coloring.

    In the other direction, let $C_1', C_2', \ldots, C_k'$ be a partition of $V(H)$ into $k$ independent sets. For each $i \in [k]$, let $C_i \subseteq V(G)$ be the union of all endpoints of the edges in $C_i'$. We now show that each $C_i$ induces a clique in $G$. Indeed, assume $u \neq v \in C_i$, and $uv \notin E(G)$. By construction of $C_i$, there are edges $e, e' \in B$ such that $e$ has $u$ as an endpoint, and $e'$ has $v$ as an endpoint.
    However, by construction of $H$, $uv \notin E(G)$ implies that $ee'$ is an edge in $H$, contradicting the assumption that $C_i'$ is an independent set in $H$. Therefore, $C_1$, $C_2$, \ldots, $C_k$ is a collection of cliques in $G$, and each edge of $B$ belongs to at least one of them. Therefore, $(G, B, k)$ is a yes-instance of \probAECCshort.
\end{proof}

As \textsc{$k$-Coloring} is polynomial-time solvable when $k \leq 2$, we immediately get the following corollary.

\begin{corollary}\label{thm:annotated_k_two}
	\probAECCshort is polynomial-time solvable for $k \leq 2$.
\end{corollary}

We note that a polynomial-time reduction from \textsc{$k$-Coloring} or \textsc{Vertex Clique Cover} to \probAECCshort follows immediately from a reduction to \probECCshort.

\subsection{Dichotomy for \probECCISshort}\label{sec:dichotomy}
In this subsection, we prove that \probECCISshort is \classParaNP-complete when parameterized by  $k$. More precisely, we show that the problem is \classNP-complete for every $k\geq 2$ on instances where the graph is perfect without isolated vertices. 
We complement this result by proving that for $k=0$ or $1$, the problem is in \classP. We start with obtaining the following result for \probAECCshort. The arguments are similar to the arguments of Orlin~\cite{orlin:cc}.

\begin{lemma}\label{thm:ann-hard}
	\probAECCshort is \classNP-complete on co-bipartite graphs for every $k\geq 3$. Furthermore, the hardness holds for the case when $B$ is a perfect matching between the bipartition sets.
\end{lemma}

\begin{proof}
It is well-known that the classical \textsc{$k$-Coloring} problem is \classNP-complete for every $k\geq 1$~\cite{garey-johnson}. This immediately implies the lower bound for the dual \probVCC, that is, for 		every $k\geq 3$, it is \classNP-complete to decide, given a graph $G$, whether there is a family of at most $k$ cliques of $G$ such that each vertex of $G$ belongs to at least one clique of the family.  We reduce from  \probVCC.   

Let $(G,k)$ be an instance of \probVCC and let $V(G)=\{v_1,\ldots,v_n\}$. We construct the co-bipartite graph $G'$ and the set of edges $B$ as follows.
\begin{itemize}
\item Construct two copies $V_1=\{v_1^{(1)},\ldots,v_n^{(1)}\}$ and $V_2=\{v_1^{(2)},\ldots,v_n^{(2)}\}$ of $V(G)$ and make them cliques.
\item For each $i\in\{1,\ldots,n\}$, construct and edge $v_i^{(1)}v_i^{(2)}$ and set $B=\{v_i^{(1)}v_i^{(2)}\mid 1\leq i\leq n\}$.
\item For each $v_iv_j\in E(G)$ for $i,j\in \{1,\ldots,n\}$, construct edges $v_i^{(1)}v_j^{(2)}$ and $v_j^{(1)}v_i^{(2)}$.
\end{itemize}
Note that $G'$ is a co-bipartite graphs where $V_1$ and $V_2$ form a bipartition of the vertex set.
We claim that the vertices of $V(G)$ can be covered by at most $k$ cliques if and only if the edges of $B$ can be covered by at most $k$ cliques in $G'$.

Assume that $C_1,\ldots,C_k$ are cliques of $G$ with the property that  for every $i\in \{1,\ldots,n\}$, there is $j\in\{1,\ldots,k\}$ such that $v_i\in C_j$. For every $j\in\{1,\ldots,k\}$, we define 
$C_j'=\{v_i^{(1)}\mid v_i\in C_j\}\cup\{v_i^{(2)}\mid v_i\in C_j\}$. By the construction of $G'$, we have that each $C_j'$ is a clique of $G'$. Because each $v_i$ is included in some $C_j$, we obtain that $v_i^{(1)},v_i^{(2)}\in C_j' $. Thus, $C_1',\ldots,C_k'$ cover $B$.

For the opposite direction, assume that there are $k$ cliques $C_1',\ldots,C_k'$ of $G'$ with the property that for every $i\in\{1,\ldots,n\}$, there is $j\in\{1,\ldots,k\}$ such that $v_i^{(1)},v_i^{(2)}\in C_j'$. We assume without loss of generality that $C_j'$ are inclusion minimal cliques covering $B$. Then for each $j\in\{1,\ldots,k\}$, a vertex $v_i^{(1)}\in C_j'$ if and only if $v_i^{(2)}\in C_j'$. The construction of $G'$ implies that for each $j\in\{1,\ldots,k\}$, $C_j=\{v_i\mid v_i^{(1)},v_i^{(2)}\in C_j'\}$ is a clique of $G$. Since  each $v_i^{(1)},v_i^{(2)}\in C_j'$ for some $j\in\{1,\ldots,k\}$, we obtain that $C_1,\ldots,C_k$ cover $V(G)$. This concludes the proof.
\end{proof}

We use \Cref{thm:ann-hard}  to prove \Cref{thm:hard}. We restate it here.

\ECCISNPhard*

\begin{proof}  
We use \Cref{thm:ann-hard} and reduce from \probAECCshort. 

Let $(G,B,k)$ be an instance of \probAECCshort where $G$ is an $n$-vertex co-bipartite graph for $n\geq 2$. We also assume that $B$ is a perfect matching in $G$, that is, each vertex of $G$ is incident to an edge of $B$.  We set  $R=E(G)\setminus B$ and construct the graph $G'$ as follows.

\begin{itemize}
\item Construct a copy of $G$.
\item Construct a vertex $v$ and make it adjacent to each vertex of $G$.
\item For each vertex $x\in V(G)$, construct a vertex $u_x$ and make it adjacent to $x$.
\item For every edge $e=xy\in R$, construct a vertex $w_e$ and make it adjacent to $x$ and $y$.
\end{itemize}

Because $G$ is a co-bipartite graph with at least two vertices that has a perfect matching, $G$ has no isolated vertices. Then the construction implies that the same holds for $G'$.
Notice that $G'$ is a perfect graph. To see this, we use the perfect graph theorem~\cite{ChudnovskyRST06}. Trivially, a vertex of degree one does not belong to any odd hole or antihole. For any vertex $w_e$ for $e\in R$, $w_e$ does not belong to any odd hole because the neighbors of $w_e$ are adjacent, and $w_e$ does not belong to any antihole with at least 7 vertices because the degree of $w_e$ is two. Since the graph obtained from the co-bipartite graph $G$ by adding the universal vertex $v$ is co-bipartite, this graph is perfect. This immediately implies that $G'$ is perfect.
 
 Consider the set $I=\{v\}\cup\{u_x\mid x\in V(G)\}\cup\{w_e\mid e\in R\}$. By the construction of $G'$, $I$ is an independent set of size $|R|+n+1$.  We claim that $I$ is a maximum independent set. To see this, note that the vertices of $I'=\{u_x\mid x\in V(G)\}\cup\{w_e\mid e\in R\}$ are simplicial. Therefore, there is a maximum independent set of $G'$ containing $I'$. Since $v$ is the unique vertex nonadjacent  
 to the vertices of $I'$, we obtain that the size of $I$ is maximum. Thus, $\alpha(G')=|R|+n+1$. 
 
 We set $k'=k-1$ and claim the following.
 \begin{claim}
$(G,B,k)$ is a yes-instance of \probAECCshort if and only if $(G',k')$ is a yes-instance of  \probECCISshort. 
 \end{claim}
 
 \begin{claimproof}
 Suppose that there is a family $C_1,\ldots,C_k$ of $k$ cliques of $G$ such that each edge of $B$ is covered by at least one clique. For each $i\in\{1,\ldots,k\}$, we set $C_i'=C_i\cup\{v\}$. Because $v$ is adjacent to every vertex of $G$, $C_i'$ is a clique of $G'$. Recall that every vertex of $G$ is incident to an edge of $B$.  Then because the cliques $C_1,\ldots,C_k$ cover $B$, for every vertex $x\in V(G)$, there is a clique $C_i$ such that $x\in C_i$. This implies that $x,v\in C_i'$. Thus, the edges $xv$ for $x\in V(G)$ are covered by the cliques $C_1',\ldots,C_k'$. Also, these cliques cover $B$. For every vertex $x\in V(G)$, we set  $X_x=\{x,u_x\}$. Clearly, this is a clique of size two. Then the family of cliques $\{X_i\mid x\in V(G)\}$ covers every edge $xu_x$ for $x\in V(G)$. Finally, for every $e=xy\in R$, we consider the clique $Y_e=\{w_e,x,y\}$ of size three and observe that these cliques cover the edges incident to the vertices $w_e$ and the edges of $R$. Summarizing, we obtain that 
 the family of $k+n+|R|$ cliques $\{C_1',\ldots,C_k'\}\cup\{X_x\mid x\in V(G)\}\cup\{Y_e\mid e\in R\}$ cover every edge of $G'$. Thus, $G'$ admits an edge clique cover of size 
 $k+n+|R|=\alpha(G)+k'$, that is, $(G',k')$ is a yes-instance of  \probECCISshort. 
 
 For the opposite direction, assume that there is a family $\mathcal{F}$ of $\alpha(G')+k'=|R|+n+k$ cliques of $G'$ covering every edge of $G'$. Observe that for every $x\in V(G)$, the edge $xu_x$ can be covered by the unique clique $X_x=\{x,u_x\}$. This means that $\mathcal{F}$ contains $n$ cliques of this type. For every $e\in R$, $\mathcal{F}$ should include a clique $Y_e$ containing $w_e$ and at least one of the endpoints of $e$. Thus, we have at least $|R|$ cliques containing the vertices $w_e$ for $e\in R$. Notice that these cliques do not cover the edges of $B$. Therefore, $\mathcal{F}$ should contain  $r\leq (\alpha(G')+k')-n-|R|=k$  cliques $C_1',\ldots,C_r'$ covering the edges of $B$.  Furthermore, for every $i\in\{1,\ldots,r\}$, $C_i'$ is a clique of $G'[V(G)\cup\{v\}]$. Let $C_i=C_i'\setminus\{v\}$ for $i\in\{1,\ldots,r\}$ and observe that $C_i$ is a clique of $G$. Then the cliques $C_1,\ldots,C_r$ cover the edges of $B$. Thus, $(G,B,k)$ is a yes-instance of \probAECCshort. 
 \end{claimproof}
 This concludes the proof of the theorem.
 \end{proof}
 
 As a corollary of  \Cref{lemma:annotated_to_coloring} and \Cref{thm:annotated_k_two}, we obtain that when $k=0$ or $k=1$,  \probECCISshort can be solved in polynomial time. Thus, combining \Cref{thm:hard} and \Cref{cor:poly-zero-one}, we get the computational complexity dichotomy for the problem with respect to $k$.
 
 \ECCISPoly*

\begin{proof}
Given an instance $(G,k)$ of \probECCISshort for $k\leq 1$, the algorithm from \Cref{{lemma:ecc_to_annotated}} in polynomial time reduces the problem to solving at most two instances $(G',B,k')$ of
\probAECCshort where $k'\leq 2$. Then we apply the algorithm from \Cref{thm:annotated_k_two} to the obtained instances and  solve the problem in polynomial time.
\end{proof}

 Using the result of Orlin~\cite{orlin:cc} about the \classNP-completeness of the \textsc{Edge Biclique Cover} problem, it is easy to see that \probECCshort is \classNP-complete on co-bipartite graphs. This result immediately implies  that    \probECCISshort is \classNP-complete on co-bipartite graphs. We provide the proof for completeness  below.
 
 \begin{observation}\label{obs:NPc-cobip}
\probECCISshort is \classNP-complete on co-bipartite graphs.
 \end{observation}
 
 \begin{proof}
 A set of vertices of a bipartite graph is a \emph{biclique} if it induces a complete bipartite graph. We say that a biclique covers an edge if this edge is in the subgraph induced by the biclique. The task of \textsc{Edge Biclique Cover} is, given a bipartite graph $G$ and a integer $k\geq 0$, to decide whether there is a family of at most $k$ bicliques such that each edge is covered by  at least one biclique. This problem is well-known to be \classNP-complete~\cite{orlin:cc}.  
 Let $(G,k)$ be an instance of \textsc{Edge Biclique Cover} and denote by $(V_1,V_2)$ the bipartition of $V(G)$. We construct $G'$ as follows.
 
 \begin{itemize}
\item Construct a copy of $G$.
\item Make the sets of vertices $V_1$ and $V_2$ cliques.
\item Construct vertices $v_1,v_2$, make $v_1$ adjacent to the vertices of $V_1$, and make $v_2$ adjacent to the vertices of $V_2$.
\end{itemize} 
 
It is straightforward to see that $G$ is co-bipartite.  Because $v_1$ and $v_2$ are not adjacent, $\alpha(G')=2$. We show that $(G,k)$ is a yes-instance of  \textsc{Edge Biclique Cover} if and only if $(G',k')$ is a yes-instance of \probECCISshort.

Suppose that bicliques $B_1,\ldots,B_k$ cover all the edges of $G$. By the constrution of $G'$, $B_1,\ldots,B_k$ are cliques of $G'$. Then together with the cliques $V_1\cup\{v_1\}$ and $V_2\cup\{v_2\}$, we obtain $k+2=k+\alpha(G)$ cliques covering the edges of $G$. For the opposite direction, assume that there is a family $\mathcal{F}$ of $\alpha(G')+k=2+k$ cliques of $G'$ covering the edges of $G'$. Observe that $\mathcal{F}$ contains at least two distinct cliques $X_1$ and $X_2$ covering the edges incident to $v_1$ and $v_2$. These cliques do not cover any edge of $G$.  This implies that $\mathcal{F}$ contains 
$r\leq (\alpha(G')+k)-2=k$ cliques $C_1,\ldots,C_r$ covering the edges of $G$. It remains to notice that  $C_1,\ldots,C_r$ are bicliques of $G$ covering $E(G)$. Thus, $(G,k)$ is a yes-instance of   \textsc{Edge Biclique Cover}. This concludes the proof. 
\end{proof}

Because \probECCISshort is \classNP-complete on co-bipartite graphs, we have that for every integer $p\geq 2$, \classNP-complete on graphs $G$ with $\alpha(G)\leq p$. It is trivial to see that if $\alpha(G)=1$ then $G$ is a complete graph and, therefore, the edges of $G$ can be covered by a single clique. 

\Cref{thm:hard} and \Cref{obs:NPc-cobip} imply that \probECCISshort is \classParaNP-hard when parameterized by either $k$ or the independence number. Because \probECCshort  is \classFPT when parameterized by the number of cliques in a solution~\cite{GrammGHN08}, it is straightforward to make the following observation.

\begin{observation}\label{obs:fpt-k-alpha}
\probECCISshort is \classFPT when parameterized by both $k$ and the independence number of the input graph.
\end{observation}

\section{Covering Sparse Graphs}\label{sec:FPTbyomega}

In this section, we design \classFPT algorithms for \probECCISshort for several well-studied graph classes, where \probECCshort or \pname{Independent Set} remain \classNP-complete.
We build on \Cref{lemma:ecc_to_annotated} (and \Cref{cor:coloring_to_annotated_faster}), that allow us to transform parameterized algorithms for \probAECCshort (and \pname{Independent Set}) on a graph class $\mathcal{G}$, into parameterized algorithms for \probECCISshort on $\mathcal{G}$.
The only requirement for $\mathcal{G}$ is being closed under deletion of simplicial vertices.

In the first two subsections of this section, we give parameterized algorithms for \probAECCshort on $K_{r}$-free, $d$-degenerate and $H$-minor-free graphs.
In the third subsection, we combine these algorithms into corresponding algorithms for \probECCISshort, proving \Cref{cor:sparse}.
Finally, we complement these results with lower bounds based on the Exponential Time Hypothesis.

\subsection{Bounded Clique Number}

\probAECCshort on $K_r$-free graphs can alternatively be seen as \probAECCshort parameterized by $k+\omega$, where $\omega$ is the maximum clique size in $G$.
Naturally, we cannot have more than $\binom{\omega}{2}\cdot k$ edges in $B$ in a yes-instance of \probAECCshort, which hints that \probAECCshort should be \classFPT with respect to $k+\omega$.
There is, however, a crucial obstacle: computing $\omega$ is \classNP-hard and \classW{1}-hard when parameterized by $\omega$.

Fortunately, independent sets of size more than $k$ in $G$ will guarantee that $(G,B,k)$ is a no-instance.
This allows us to use the classical Ramsey's theorem~\cite{Ramsey_1930} to either conclude that the instance is negative or to find $\omega$. Formally, we use the following algorithmic result that is derived from the Ramsey number upper bound given by Erd\H{o}s and Szekeres~\cite{CM_1935__2__463_0}.

\begin{proposition}[Adaptation of a proof by Erd\H{o}s and Szekeres~\cite{CM_1935__2__463_0}]\label{lemma:ramsey_clique_or_is}	
There is a polynomial-time algorithm that, given $n$-vertex graph $G$ and two integers $p,q$ such that $n\ge \binom{p+q-2}{p-1}$, finds either a clique of size $p$ or an independent set of size $q$ in $G$.
\end{proposition}

We are ready to prove that \probAECCshort is \classFPT with respect to $k+\omega$.

\begin{lemma}\label{thm:ECCclique}
	\probAECCshort is fixed-parameter tractable with respect to the combined parameter $k+\omega$.
	The running time of the corresponding algorithm is $2^{\binom{\omega}{2}\cdot k}\cdot \polyn$.
	The value of $\omega$ needs not to be given in the input.
\end{lemma}
\begin{proof}
	We present an algorithm that finds a solution to the given instance $(G,B,k)$ of \probAECCshort.
	We assume that $G, B$ are not empty, as otherwise the instance is trivial.
	If $k \leq 2$, we use the algorithm of \Cref{thm:annotated_k_two} as a subroutine to solve the instance in polynomial time.
	We explain how the algorithm deals with the general case $k\ge 3$ in several steps below.
				
	\medskip\noindent{\bf Triangle-free check.} The algorithm first checks whether $G$ is triangle-free, that is, $\omega>2$.
	This is performed in polynomial time.
	If $G$ has $\omega=2$, then no two edges can belong to the same clique in $G$, so optimal cover of edges in $B$ consists of $|B|$ cliques.
	Hence, if $G$ turns out to be triangle-free, then the algorithm reports that $(G,B,k)$ is a yes-instance if $|B|\le k$ and no-instance otherwise.
	
	\medskip\noindent{\bf Clique or independent set.}
	If there is a vertex $v \in V(G)$ without incident edges in $B$, 
	then there exists an optimal solution where none of the cliques contain $v$.
	The algorithm removes all such vertices from $G$.
	Now any solution to $(G,B,k)$ should contain each vertex of $G$ in at least one of the cliques.
	The algorithm then finds the largest $r$ such that $n\ge\binom{k+r-1}{k}$ and $n< \binom{k+r}{k}$.
	Substituting $r$ and $k+1$ in \Cref{lemma:ramsey_clique_or_is}, the algorithm finds either an independent set of size $k+1$ in $G$ or a clique of size $r$ in $G$ in polynomial time.
	An independent set of size $k+1$ guarantees that $(G,B,k)$ is a no-instance, as it is impossible to cover each vertex of $G$ by using at most $k$ cliques.
	With this outcome, our algorithm reports a no-instance and stops.
	In the complementary case, we have that $\omega \ge r$ and $n<\binom{k+\omega}{k}$.
	
	\medskip\noindent{\bf Computing $\omega$.}
	For an integer $p \ge 3$, $\omega < p$ can be verified in time $\Oh(p^2\cdot n^p)$.
	As the next step, the algorithm computes the exact value of $\omega$ by performing these checks iteratively for consecutive integer values starting from $\max\{4,r+1\}$ up to $\omega+1$.
	This routine takes $n^{\omega+\mathcal{O}(1)}$ running time.
	To give an upper bound on the running time in terms of $k$ and $\omega$, we need the following claim.
	
	\begin{claim}\label{lemma:big_math_lemma}
		If $a, b$ are positive integers with $a,b\ge 3$, then $\binom{a+b}{b}\le 2^{\frac{ab}{2}}$.
	\end{claim}
	\begin{claimproof}
		Without loss of generality, $a\le b$.
		If $a\ge 4$ or $b \ge 6$, we have $$ab=(a-2)(b-2)+2(a+b)-4\ge 2(a+b).$$
		Then $\binom{a+b}{b}\le 2^{a+b}\le 2^{\frac{ab}{2}}$ holds.
		The only values of $(a,b)$ that do not fit within these lower bounds are $(3,3)$, $(3,4)$, $(3,5)$.
		All three of these value pairs satisfy the inequality.
	\end{claimproof}
	
	By substituting $k,\omega\ge 3$ in \Cref{lemma:big_math_lemma}, we obtain
	\[
	n^{\omega -1}< \binom{k+\omega}{\omega}^{\omega-1}<\left(2^{\frac{k\omega}{2}}\right)^{\omega-1}=2^{\binom{\omega}{2}\cdot k}.
	\]
	Hence, it takes time $2^{\binom{\omega}{2}\cdot k}\cdot\polyn$ for the algorithm to compute the exact value of $\omega$.
	
	\medskip\noindent{\bf Solution via set cover.}
	No clique in $G$ can cover more than $\binom{\omega}{2}$ edges in $B$.
	Thus, if $|B|> \binom{\omega}{2}\cdot k$, the algorithm reports correctly that $(G,B,k)$ is a no-instance.
	Otherwise, the algorithm reduces $(G,B,k)$ to an instance $(B, \mathcal{F}, k)$ of \textsc{Set Cover}.
	In this instance, $B$ is treated as the universe set.
	Subset family $\mathcal{F}$ consists of subsets of $B$, and a subset $B'\subset B$ belongs to $\mathcal{F}$ if and only if all edges of $B'$ belong simultaneously to the same clique in $G$.
	Clearly, $(G,B,k)$ is a yes-instance if and only if $B$ can be covered by at most $k$ sets from $\mathcal{F}$, i.e.\ $(B,\mathcal{F},k)$ is a yes-instance of \textsc{Set Cover}.

	The \textsc{Set Cover} instance is constructed in time $2^{|B|}\cdot\polyn$ by the algorithm.
	Using the celebrated subset convolution algorithm for \textsc{Set Cover} due to Bj\"{o}rklund, Husfeldt and Koivisto  \cite{Bjrklund2009} as a subroutine, our algorithm finally decides $(B,\mathcal{F},k)$ within the same $2^{|B|}\cdot\polyn$ running time bound. The overall running time is therefore dominated by $2^{\binom{\omega}{2}\cdot k}\cdot\polyn$.
\end{proof}

\subsection{Degenerate and Minor-free Graphs}

We move on to the setting of low degeneracy.
Recall that \probECCshort is \classNP-complete on graphs of maximum degree $6$ \cite{hoover1992complexity} and planar graphs \cite{chang2001tree}, hence on $5$-degenerate graphs.
We present an algorithm for \probAECCshort that runs in $2^{\Oh(dk)}\cdot\polyn$ time on $d$-degenerate graphs.
Notably, on $2$-degenerate graphs its running time is polynomial. 
 
\begin{lemma}\label{lemma:aecc_on_degenerate}
	For $d\ge 2$, \probAECCshort on $d$-degenerate graphs admits an algorithm running in $c_{d-1}^{k}\cdot \polyn$ time, where $c_{p}\le 3^{p/3}$ is the maximum  possible number of maximal cliques in a graph with $p$ vertices.
\end{lemma}
\begin{proof}
	We present a simple recursive branching algorithm.
	Given an instance $(G,B,k)$ of \probAECCshort, the algorithm first removes all vertices from $G$ that are not endpoints of $B$.
	If $k=0$, the algorithm reports that $(G,B,k)$ is a yes-instance if $G$ has no vertices and a no-instance otherwise, and stops.
	
	If $k\ge 1$, the algorithm takes a vertex $v$ of minimum degree in $G$, $\deg_G(v)\le d$.
	Then it takes any $u \in N_G(v)$ such that $uv\in B$.
	Note that any clique that contains $uv$ has all of its other vertices in $N_G(u)\cap N_G(v)$, denote this set by $S$, $|S|\le d-1$.
	Any \emph{maximal} clique in $G$ that contains $uv$ is also a maximal clique in $G[S]$.
	Hence, there exists an optimal solution to $(G,B,k)$ that contains a clique with the vertex set $C\cup\{u,v\}$ for some maximal clique $C$ in $G[S]$.
	This implies the only branching rule of the algorithm: For each maximal clique $C$ in $G[S]$, run the algorithm recursively on the instance $(G,B-E(G[C\cup \{u,v\}]),k-1)$.
	
	The algorithm description is complete, its correctness follows from the discussion above.
	The number of branches in a single application of the branching rule is bounded by the maximum possible number of maximal cliques in a $(d-1)$-vertex graph, which we denote by $c_{d-1}$.
	By Moon and Moser's result \cite{Moon1965}, $c_p$ equals $3^{p/3}$, $4\cdot 3^{(p-4)/3}$ or $2\cdot 3^{(p-2)/3}$ for the remainders $0,1,2$ of $p$ modulo $3$ respectively.
	The inequality $c_p\le 3^{p/3}$ follows from $2<3^{2/3}$.
	Clearly, recursion depth is bounded by $k$, so the overall running time is bounded by $c_{d-1}^k\cdot\polyn$.
\end{proof}

Since planar graphs are $5$-degenerate, $4^k\cdot\polyn$ algorithm for \probAECCshort on planar graphs follows.
We show that for planar graphs, somewhat expectedly, this can be improved to a subexponential parameterized algorithm with $2^{\Oh(\sqrt{k})}\cdot \polyn$ running time.
For this, we could exploit two specific properties of planar graphs.
First, treewidth of a planar $n$-vertex graph is $\Oh(\sqrt{n})$. 
Second, any planar graph $G$ satisfies $\omega(G)\le 4$, and an instance with $n>4k$ is a trivial no-instance, as soon as every vertex is incident to an edge in $B$.
Consequently, the treewidth of $G$ is $\Oh(\sqrt{k})$ in non-trivial instances of \probAECCshort.
Then, dynamic programming over a tree decomposition of $G$ can be applied.

Scenario of such dynamic programming has been considered in the literature.
In \cite{BlanchetteKV12}, Blanchette, Kim and Vetta proved that \probECC can be solved in time $2^{\Oh(\rho)}\cdot\polyn$, given a nice tree decomposition of $G$ where maximum number of edges induced by a bag is bounded by $\rho$.
This dynamic programming can be easily adapted to work for \probAECCshort.
For tree decompositions of width $t$ of a $d$-degenerate graph, $\rho\le dt$.
We formulate this algorithmic result in the following proposition.

\begin{proposition}[Blanchette, Kim and Vetta \cite{BlanchetteKV12}]\label{prop:ecc_tw_dp}
	\probAECCshort admits an algorithm that, given a nice tree decomposition  of $G$ of width $t$,  works in $2^{\Oh(dt)}\cdot\polyn$ running time, where $d$ is the degeneracy of $G$.
\end{proposition}

While the discussion above mainly gives intuition and technical details for planar graphs, next we present in detail a more general algorithm that deals with $H$-minor-free graphs.
For some computable function $f$ and any finite graph family $H$, the algorithm works in time $f(H)^{\sqrt{k}}\cdot\polyn$ on $H$-minor-free graphs.
Note that $H$ remains unknown to the algorithm, and the algorithm exploits properties of $H$-minor-free graphs implicitly.

\begin{lemma}\label{lemma:aecc_on_h_minor_free}
	\probAECCshort admits an algorithm that runs in $f(H)^{\sqrt{k}}\cdot \polyn$ time on $H$-minor-free graphs, for some computable function $f$.
	The algorithm does not receive $H$ in the input.
\end{lemma}
\begin{proof}
	We start with the description of the algorithm.
	First, the algorithm removes all vertices from $G$ that are not incident to $B$.
	Second, the algorithm computes the degeneracy $d$ of $G$ in polynomial time in the standard manner.
	Next, if there are more than $(d+1)\cdot k$ vertices in $G$, the algorithm terminates and reports that $(G,B,k)$ is a no-instance.
	Otherwise it computes a nice tree decomposition $\mathcal{T}$ of $G$ of width $t$ that is at most $2\cdot \operatorname{tw}(G)+1$, using the $2$-approximation algorithm of Korhonen \cite{Korhonen2022}.
	Finally, our algorithm employs the dynamic programming algorithm of \Cref{prop:ecc_tw_dp} over $\mathcal{T}$ to find the optimal solution.
	
	Correctness of the algorithm is straightforward.
	We now prove that its running time is subexponential for fixed $H$, given that $G$ is $H$-minor-free.
	If the algorithm terminates at its third step, its running time is polynomial.
	Then observe that degeneracy $d$ of $G$ is bounded by $\Oh(h/\sqrt{\log h})$, where $h$ is the maximum size of a graph in $H$ \cite{Kostochka1984}.
	It follows that before the fourth step of the algorithm, $n\le (d+1)k\le g_1(H)\cdot k$ for a computable function $g_1$.
	Treewidth of $G$ is bounded by $g_2(H)\cdot \sqrt{n}$ \cite{Grohe2003}, for a computable function $g_2$.
	Hence $\mathcal{T}$ has width $t\le g_3(H)\cdot \sqrt{k}$ for a computable function $g_3$  after the fourth step of the algorithm.
	The $2$-approximation algorithm for treewidth works in $2^{\Oh(t)}\cdot \polyn$ time, and that gives an upper bound for the running time of the fourth step of the algorithm.
 	The final step requires $2^{\Oh(dt)}\cdot\polyn$ computation time.
 	
 	The final upper bound dominates all previous ones.
 	Since $dt\le g_2(H)\cdot g_3(H)\cdot \sqrt{k}$, overall running time of the algorithm is bounded by $f(H)^{\sqrt{k}}\cdot\polyn$, and order of $f$ is exponential in $g_2(H)\cdot g_3(H)$.
\end{proof}

\subsection{Combining the Results}

By combining  \Cref{lemma:ecc_to_annotated} or  \Cref{cor:coloring_to_annotated_faster} with \Cref{thm:ECCclique}, \Cref{lemma:aecc_on_degenerate}, and \Cref{lemma:aecc_on_h_minor_free},
we obtain 
\Cref{cor:sparse}, which we restate here. 

\ECCISsparse*

\begin{proof}
	Recall that \Cref{lemma:ecc_to_annotated} provides a reduction from \probECCISshort to at most $2k$ instances of \probAECCshort on induced subgraphs of $G$, each with parameter bounded by $2k$, that works in $4^k\cdot\polyn$ time.
	To obtain the first algorithm of the current theorem, pipeline the reduction of \Cref{lemma:ecc_to_annotated} with the algorithm of \Cref{thm:ECCclique} for \probAECCshort parameterized by $k+\omega$.
	On each instance the algorithm works in $2^{\binom{\omega}{2}\cdot 2k}\cdot\polyn$ time.
	This matches the bound given in the theorem statement and dominates the running time of the reduction in graphs with $\omega\ge 2$; other graphs are trivial instances.	
	
	For graphs of degeneracy $d\ge 3$, combine \Cref{lemma:ecc_to_annotated} with the algorithm of \Cref{lemma:aecc_on_degenerate} for parameterization by $k+d$.
	Since $c_{d-1}^2\ge 4$ for $d\ge 3$, $c_{d-1}^{2k}\cdot\polyn$ dominates the reduction running time, and running time of the combined algorithm is upper-bounded by $c_{d-1}^{2k}\cdot\polyn$.
	To obtain upper bound in the theorem statement, use $c_{d-1}\le 3^{(d-1)/3}$ with $3^{2/3}<2.081$.
	
	The desired running time for the remaining two algorithms is smaller than the running time of reduction in \Cref{lemma:ecc_to_annotated}, so we pipeline with \Cref{cor:coloring_to_annotated_faster} instead. Recall that \Cref{cor:coloring_to_annotated_faster} requires to compute the independence number on the constructed instances of \probAECCshort.
	
	For graphs of degeneracy $2$, \Cref{lemma:aecc_on_degenerate} provides a polynomial-time algorithm.
	However, we cannot efficiently compute $\alpha(G)$ on $2$-degenerate graphs under $\classP\neq\classNP$, since \pname{Independent Set} is \classNP-hard on cubic graphs.
	On the other hand, there is a straightforward $(1,2)$-branching\footnote{For a vertex $v$ of degree $2$ in $G$, there is an optimal independent set in $G$ that either contains $v$ or contains $N_G[v]$.}  for $2$-degenerate graphs that gives a $1.619^k\cdot\polyn$ algorithm for \pname{Independent Set}.
	Combining \Cref{cor:coloring_to_annotated_faster} with \Cref{lemma:aecc_on_degenerate} and this branching algorithm for \pname{Independent Set}, we obtain the third algorithm for \probECCISshort.
	
	To obtain a subexponential algorithm for \probECCISshort on $H$-minor-free graphs, we require subexponential algorithm for \pname{Independent Set} on $H$-minor-free graphs.
	While this is well-known, we give a simple proof here for completeness.
	The proof is an adaptation of the proof of \Cref{lemma:aecc_on_h_minor_free}.
	We show it in the following claim.
	
	\begin{claim}\label{claim:is_h_minor_free}
		\pname{Independent Set} on $H$-minor-free graphs is solvable in $f(H)^{\sqrt{k}}\cdot\polyn$ time for some computable function $f$.
	\end{claim}
	\begin{claimproof}
		Adapt the proof of \Cref{lemma:aecc_on_h_minor_free}  with following changes.
		The algorithm takes $(G,k)$ as an input and reports ``yes'' if $G$ has an independent set of size $k$ and ``no'' otherwise.
		There is no step concerning $B$ anymore.
		Then degeneracy is computed in polynomial time, and if $n>(d+1)\cdot k$, $G$ has independent set of size $k+1$, so the algorithm terminates and reports ``yes''.
		Otherwise, the algorithm proceeds and finds $\mathcal{T}$ in exactly the same way.
		In its final step, the algorithm employs dynamic programming over $\mathcal{T}$ that works in $2^{t}\cdot\polyn$ time (see, e.g.\ \cite{CyganFKLMPPS15}).
	\end{claimproof}
	
	Finally, \Cref{cor:coloring_to_annotated_faster} is pipelined with algorithms of \Cref{lemma:aecc_on_h_minor_free} and \Cref{claim:is_h_minor_free} to obtain the algorithm for \probECCISshort that runs in desired time on $H$-minor-free graphs.
	This finishes the proof.
\end{proof}

\subsection{Corresponding Lower Bounds}\label{sec:lowerbounds}

We support \Cref{cor:sparse} with corresponding conditional lower bounds.
We combine well-known reductions and lower bounds in the proof of the following proposition.
It shows that, under the ETH, the dependency on $k$ cannot be improved in neither of the algorithms; additionaly, the dependency on $d$ is almost optimal.

\begin{proposition}
	Under the Exponential Time Hypothesis, algorithms with the following running times do not exist:
	\begin{itemize}
		\item $2^{o(k)}\cdot \polyn$ on $2$-degenerate graphs for \probECCISshort ($5$-degenerate graphs for \probECCshort);
		\item $2^{o(d/\log d)\cdot k}\cdot\polyn$ for both \probECCshort and \probECCISshort, where $d=\operatorname{dg}(G)$;
		\item $2^{o(\sqrt{k})}\cdot\polyn$ on planar graphs for both \probECCshort and \probECCISshort.
	\end{itemize}
\end{proposition}
\begin{proof}
	First note that lower bounds for \probECCISshort follow from corresponding lower bounds for \probECCshort similarly to \Cref{obs:isolhard}, since all considered graph classes are closed under the addition of pendent vertices.
	
	Under the ETH, both \pname{Vertex Cover} and \pname{Independent Set} do not admit $2^{o(n+m)}\cdot\polyn$ algorithms on connected graphs of maximum degree three \cite{ImpagliazzoPZ01, Garey1974}.
	For planar graphs of maximum degree three, there are no $2^{o(\sqrt{n+m})}\cdot \polyn$ algorithms under ETH \cite{Garey1974,Garey1977}.
	
	In \cite{hoover1992complexity}, 
	Hoover showed that an instance $(G,k)$ of \textsc{Vertex Cover}, where $G$ is connected and has maximum degree three, can be reduced to an instance $(G',k')$ of \probECCshort, where $G'$ is connected and has maximum degree six.
	Moreover, this reduction preserves planarity and satisfies $n'+m'=\Oh(n+m)$, where $n,m,n',m'$ is the number of vertices and edges in $G$ and $G'$ correspondingly.
	
	The lower bound for planar graphs is now straightforward.
	Indeed, if \probECCshort admits an algorithm with running time $2^{o(\sqrt{k})}\cdot\polyn$, combine this algorithm with a reduction that reduces an instance $(G,k)$ of \pname{Vertex Cover} to an instance $(G',k')$ of \probECCshort, where $G$ and $G'$ are both planar.
	Since $k'\le |E(G')|=\Oh(n+m)$, the combination gives a $2^{o(\sqrt{n+m})}\cdot\polyn$ running time algorithm for \pname{Vertex Cover} on planar graphs that refutes ETH.
	
	We now prove that no subexponential algorithm exists for \probECCshort on $5$-degenerate graphs under the ETH.
	Assume that \probECCshort on $5$-degenerate graphs admits an algorithm $\mathcal{A}$ with $2^{o(k)}\cdot\polyn$ running time.
	To solve an instance $(G,k)$ of \pname{Vertex Cover}, where $G$ is connected with $\Delta(G)\le 3$, reduce it to an instance $(G',k')$ of \probECCshort where $G'$ is connected and $\Delta(G')\le 6$.
	If $G'$ is not $6$-regular, then it is $5$-degenerate, and $(G',k')$ can be solved via $\mathcal{A}$ in $2^{o(k')}\cdot (n')^{\Oh(1)}$ time.
	If $G'$ is $6$-regular, then we have to reduce it to several instances of \probECCshort.
	To do this, take an arbitrary edge $e$ of $G'$; it belongs to $\mathcal{O}(1)$ cliques in $G'$, since $\Delta(G')=6$.
	There are constantly many sets of edges containing $e$ that belong to the same clique in $G'$.
	For each such edge set $X$, apply $\mathcal{A}$ to $(G'-X, k'-1)$, given that $G'-X$ is $5$-degenerate.
	If $(G',k')$ is a yes-instance, then $(G'-X,k'-1)$ is a yes-instance for at least one choice of $X$, and vice versa.
	We obtained an algorithm resolving an instance $(G,k)$ of \pname{Vertex Cover} via $\mathcal{O}(1)$ calls to $\mathcal{A}$.
	As soon as $k'\le |E(G')|=m'=\mathcal{O}(n+m)$, the running time is subexponential in $n+m$.
	Thus existence of $\mathcal{A}$ contradicts the ETH.
	
	While lower bound for \probECCISshort on $5$-degenerate graphs follows, we now derive a stronger lower bound for \probECCISshort.
	\Cref{lemma:aecc_on_degenerate} shows that \probECCshort admits polynomial-time algorithm for $2$-degenerate graphs.
	On the other hand, under the ETH, \pname{Independent Set} does not admit subexponential algorithms for connected graphs of maximum degree three and, consequently, $2$-degenerate graphs.
	Assume now that an algorithm running in $2^{o(k)}\cdot\polyn$ time exists for \probECCISshort on $2$-degenerate graphs.
	Then, given an instance $(G,k)$ of \pname{Independent Set}, where $G$ is $2$-degenerate, evaluate $\ecc(G)$ in polynomial time using the algorithm of \Cref{lemma:aecc_on_degenerate}.
	Since $\alpha(G)\ge k$ is equivalent to $\alpha(G)+\ecc(G)-k\ge \ecc(G)$, $G$ has independent set of size $k$ if and only if $(G,\ecc(G)-k)$ is a yes-instance of \probECCISshort.
	Since $\ecc(G)\le m$, this is evaluated in $2^{o(m)}\cdot\polyn$ time under the initial assumption.
	It follows that $2^{o(k)}\cdot\polyn$-time algorithm for \probECCISshort on $2$-degenerate graphs also violates ETH.
	
	For the remaining lower bound we use the reduction of Cygan, Pilipczuk and Pilipczuk \cite{CyganPP16}.
	They showed that an instance of \pname{$3$-SAT} with $n$ variables and $m$ clauses can be transformed into equivalent instance $(G',k')$ of \probECCshort, where $|V(G')|= \Oh(n+m)$ and $k'=\Oh(\log n)$.
	Then an algorithm that decides $(G',k')$ in time $2^{o(d/\log d)\cdot k'}$, where $d<|V(G')|$ is degeneracy of $G'$, when combined with the reduction, gives an algorithm with running time $2^{o(n+m)}$ for \pname{$3$-SAT}.
\end{proof}


\end{document}